\documentclass[draftclsnofoot]{IEEEtran}
\onecolumn
\usepackage{latexsym}
\usepackage{cite}
\usepackage{makeidx}
\makeindex
\usepackage{amssymb}
\usepackage{amsmath}
\usepackage[dvips]{graphics}
\usepackage{graphicx}
\usepackage{epsfig}
\usepackage{color}
\usepackage{psfrag}
 \allowdisplaybreaks
\usepackage{subfloat}
\usepackage{subfig}
\usepackage{setspace}
\usepackage{algorithm}
\usepackage{algpseudocode}
\usepackage{theorem}
\usepackage{textcomp}
\usepackage{helvet}
\theoremheaderfont{\upshape\slshape} \theoremstyle{plain}

\theorembodyfont{\rmfamily}
 
\theoremheaderfont{\itshape}

\newcommand{\ex}[2]{\mathbb{E}_{{#1}}\left[{#2}\right]}

\newtheorem{definition}{Definition}

\newtheorem{theorem}{Theorem}

\newtheorem{lemma}{Lemma} 
\newtheorem{corollary}{Corollary}
\title{Computable Bounds for Rate Distortion with Feed-Forward for Stationary and Ergodic Sources}

\author{Iddo Naiss and Haim Permuter
\thanks{Iddo Naiss and Haim Permuter are with the Department of Electrical and Computer Engineering,
 Ben-Gurion University of the Negev, Beer-Sheva, Israel. Emails: naiss@bgu.ac.il, haimp@bgu.ac.il.
}}
\begin{document}
\maketitle
\begin{abstract}
In this paper we consider the rate distortion problem of discrete-time, ergodic, and stationary sources with
feed forward at the receiver. We derive a sequence of achievable and computable rates that
converge to the feed-forward rate distortion. We show that, for ergodic and stationary sources, the rate
\begin{align}
R_n(D)=\frac{1}{n}\min I(\hat{X}^n\rightarrow X^n)\nonumber
\end{align}
is achievable for any $n$, where the minimization is taken over the transition conditioning probability
$p(\hat{x}^n|x^n)$ such that $\ex{}{d(X^n,\hat{X}^n)}\leq D$. The limit of $R_n(D)$ exists and is the feed-forward rate distortion.
We follow Gallager's proof where there is no feed-forward and, with appropriate modification, obtain our result.
We provide an algorithm for calculating $R_n(D)$ using the alternating minimization procedure, and present several numerical examples. We also present a dual form for the optimization of $R_n(D)$, and transform it into a geometric programming problem.
\end{abstract}


\begin{section}{Introduction}\label{SecIntro}
The rate distortion function for memoryless sources is well known and was given by Shannon in his seminal work\cite{Shannon}. Shannon\cite{Shannon} showed that the rate distortion function is the minimum of mutual information between the source $X$ and the reconstruction $\hat{X}$, where the minimization is over transition probabilities $p(\hat{x}|x)$ such that the distortion constraint is satisfied, i.e., $\ex{}{d(X,\hat{X})}\leq D$. In the case where the source is stationary and ergodic, Gallager\cite{Gal} showed  that the rate distortion is the limit of the following sequence of rates. Each member of the sequence is the $n$th order rate distortion function, which is the solution of the following minimization problem
\begin{align}
\frac{1}{n}\min I(X^n;\hat{X}^n).\nonumber
\end{align}
The minimization is over all conditional probabilities $p(\hat{x}^n|x^n)$ such that the distortion constraint is satisfied, i.e., $\ex{}{d(X^n,\hat{X}^n)}\leq D$. Gallager showed that the limit of the sequence $\frac{1}{n}\min I(X^n;\hat{X}^n)$ exists and is equal to the infimum of the sequence.

The problem of source coding with feed-forward was introduced by Weissman and Merhav\cite{WeisMer} and by Venataramanan and Pradhan\cite{VenPra}, and
is depicted in Fig. \ref{RDprob}.
\begin{figure}[h]{
  \psfrag{X}[][][1]{$X^n$}\psfrag{D}[][][1]{Decoder}\psfrag{H}[][][1]{$\ \ \ \ \ \ \ \hat{X}_n(T,X^{n-s})$}
  \psfrag{W}[][][1]{$\ \ \ \ \ \ \ \ \ \ \ \ \ \ \ \ \ \ \ T(X^n)\in\{1,2,...,2^{nR}\}$}
  \psfrag{F}[][][1]{Delay $s$}\psfrag{E}[][][1]{Encoder}\psfrag{B}[][][1]{Delay $s$}\psfrag{A}[][][0.8]{$X^{n-s}$}
 \centerline{ \includegraphics[width=10cm]{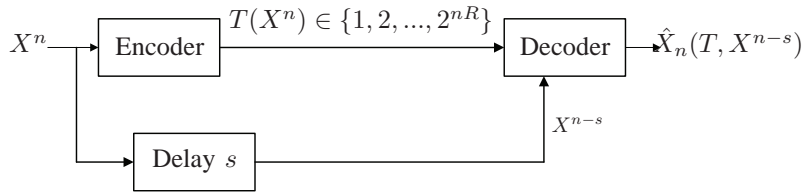}}
 \caption{Source coding with feed-forward: the decoder knows the source with delay $s$, and needs to reconstruct the source within the constraint $\ex{}{d(X^n,\hat{X}^n)}\leq D$.}
 \label{RDprob}
}\end{figure}
Weissman and Merhav\cite{WeisMer} named the problem Competitive Predictions. In their work, they defined a set of functions that predict the following $X_i$ given the previous $X^{i-1}$. After defining the \textit{loss function} between $X_i$ and the prediction, the objective was to minimizing the expected loss over all sets of predictors of size $M$. An important result in \cite{WeisMer} is that in the case where the innovation process $W_i=X^i-F_i(X^{i-1})$ is i.i.d. the distortion-rate with feed-forward function is the same as the distortion-rate function of $W_i$, where there is no feed-forward. In particular, if $X_i$ is an i.i.d. process, then $W_i=X_i$ and thus the distortion-rate with feed-forward for the source $X_i$ is the same as if there is no feed-forward.

Venkataramana and Pradhan\cite{VenPra} gave an explicit definition of the rate distortion feed-forward for an arbitrary normalized distortion function and a general source.
Their goal was to provide the rate $R$ of a source given a distortion $D$ using causal conditioning and directed information.
The source of information is modeled as the process $\{\hat{X}_n\}$ and is encoded in blocks of length $n$ into a message $T\in\{1,2,...,2^{nR}\}$. The message $T$ (after $n$ time units) is sent to the decoder that has to reconstruct the process $\{X_n\}$ using the message $T$ and causal information of the source with some delay $s$ as in Fig. \ref{RDprob}.

For that purpose, Venkataramanan and Pradhan\cite{VenPra} defined the measures
\begin{align}
\overline{I}(\hat{X}\rightarrow X)&=\limsup_{inprob}\frac{1}{n}\log\frac{p(X^n,\hat{X}^n)}{p(\hat{X}^n||X^{n-s})p(X^n)},\nonumber
\end{align}
and
\begin{align}
\underline{I}(\hat{X}\rightarrow X)&=\liminf_{inprob}\frac{1}{n}\log\frac{p(X^n,\hat{X}^n)}{p(\hat{X}^n||X^{n-s})p(X^n)}.\nonumber
\end{align}
The limsup in probability of a sequence of random variables $\{X_n\}$ is defined as the smallest extended real number $\alpha$ such that $\forall\epsilon>0$,
\begin{align}
\lim_{n\to\infty}\Pr[X_n\geq\alpha+\epsilon]=0,\nonumber
\end{align}
and the liminf in probability is the largest extended real number $\beta$ such that $\forall\epsilon>0$,
\begin{align}
\lim_{n\to\infty}\Pr[X_n\leq\beta-\epsilon]=0.\nonumber
\end{align}

The main result in \cite{VenPra} is that for a general source $\{X_n\}$ and distortion $D$, the rate distortion with feed-forward $R(D)$ is given by
\begin{align}
R(D)=\inf_{{\cal{P}}}\overline{I}(\hat{X}\rightarrow X),\nonumber
\end{align}
where the infimum is evaluated over the set ${\cal{P}}$ of probabilities $\{p(\hat{x}^n|x^n)\}_{n\geq1}$ that satisfy the distortion constraint. Moreover, if
\begin{align}
\overline{I}(\hat{X}\rightarrow X)=\underline{I}(\hat{X}\rightarrow X)\nonumber,
\end{align}
Venkataramana and Pradhan showed in \cite{VenPra}, that
\begin{align}
R(D)=\inf_{{\cal{P}}}\lim_{n\rightarrow\infty}\frac{1}{n}I(\hat{X}^n\rightarrow X^n).\nonumber
\end{align}

The work of Venkataramanan and Pradhan has made a significant contribution since it gives a multi-letter characteristic for the rate distortion function with feed-forward. In \cite{VenPra2}, they evaluated these formulas for a stock-market example and provided an analytical expression for the rate distortion function. However, these types of formulas are still very hard to evaluate for the general case. In this paper we show that assuming ergodicity and stationarity of the source, the rate distortion function with feed-forward and delay $s=1$ is upper bounded by $R_n(D)$, where
\begin{align}
R_n(D)=\frac{1}{n}\min_{p(\hat{x}^n|x^n):\ex{}{d(X^n,\hat{X}^n)}\leq D} I(\hat{X}^n\rightarrow X^n).\label{Rlim}
\end{align}
We further show that the limit of the sequence $\{R_n(D)\}$ exists, is equal to $\inf_n R_n(D)$, and is the rate distortion feed-forward function $R(D)$. These expressions for $R_n(D)$ are computable using a Blahut-Arimoto-type algorithm or using geometric programming, as demonstrated here.

In most models with causal constraints, such as feedback channels or feed-forward rate distortion, the causal conditioning probability, as well as the directed information characterizes the fundamental limits. In order to address these models, the causal conditioning probability was introduced by Massey\cite{Massey} and Kramer\cite{Kramer2} and is defined as
\begin{align}
p(\hat{x}^n||x^{n-s})=\prod_{i=1}^np(\hat{x}_i|\hat{x}^{i-1},x^{i-s}).\label{CausalPr}
\end{align}
The difference between regular and causal conditioning is that in causal conditioning the dependence of $\hat{x}_i$ on future $x_j$ is not taken into account. Following the causal conditioning probability,
Massey \cite{Massey} (who was inspired by Marko's work \cite{Marko73} on Bidirectional Communication) introduced the directed information, defined as
\begin{align}
I(\hat{X}^n\rightarrow X^n)&\triangleq H(X^n)-H(X^n||\hat{X}^n)\nonumber\\
&=\sum_{i=1}^nI(\hat{X}^i;X_i|X^{i-1}).\nonumber
\end{align}

The directed information was used by Tatikonda and Mitter\cite{TatiMit}, Permuter, Weissman, and Goldsmith \cite{PermuterWeissmanGoldsmith}, and Kim \cite{Kim} to characterize the point-to-point channel capacity with feedback. It is shown that the capacity of such channels is characterized by the maximization of the directed information over the input probability $p(x^n)$. In a previous paper\cite{NaissPermuter1}, we used these results and obtained bounds to estimate the feedback channel capacity using a Blahut-Arimoto-type algorithm (BAA) for finding the global optimum of the directed information.

The main contribution of this work lies in extending the achievability proof given by Gallager in \cite{Gal} to the case where feed-forward with delay $s=1$ exists. The extension is done by using the causal conditioning distribution, $p(\hat{x}^n||x^{n-s})$, rather than the regular reconstruction distribution $p(\hat{x}^n)$, in order to construct the codebook. The proof given is for $s=1$, but can be extended straightforwardly to any delay $s\geq1$. The difficulty in this modification is that while in \cite{Gal} the codebook was an ensemble of sequences (code words) from the reconstruction alphabet using $p(\hat{x}^n)$, our codebook is an ensemble of code trees using $p(\hat{x}^n||x^{n-s})$. This induced a major problem while showing that the probability of error is small, as discussed in Section \ref{SecAchieve}. These difficulties were overcome by appropriate modification to Gallager's proofs.

Another contribution of this paper is the development of two optimization methods for obtaining $R_n(D)$; a BA-type algorithm and a geometric programming(GP) form. The GP form is given as a maximization problem, which can be solved using standard convex optimization methods. Further, this maximization problem gives us a lower bound to the rate distortion with feed-forward, which helps us decide when to terminate the algorithm.

The remainder of the paper is organized as follows. In Section \ref{SecProbRes} we describe the problem model, provide the operational definition of the rate distortion function with feed-forward, and state our main theorems.
In Section \ref{SecAchieve} we show that $R_n(D)$ is an achievable rate for all $n$ and any distortion $D$, and in Section \ref{SecOperational} we show that the limit of $R_n(D)$ exists and is equal to the operational rate distortion function.
In Section \ref{SecalgGp} we present an alternative optimization problem for $R_n(D)$ in a standard geometric programming form that can be solved numerically using convex optimization tools. In Section \ref{secalgs} we give a description of the BAA for calculating $R_n(D)$ and present the algorithm's complexity and the memory required, and in Section \ref{SecDer} we derive the BAA and prove its convergence to the optimum value. Numerical examples are given in Section \ref{SecEx} to illustrate the performance of the suggested algorithms.
\end{section}

\begin{section}{Problem Statement and Main Results}\label{SecProbRes}
In this section we present notation, describe the problem model and summarize the main results of the paper.
We first state the definitions of a few quantities that we use in our coding theorems. We denote by $X_1^n$ the vector $(X_1,X_2,...X_n)$. Usually we use the notation $X^n=X_1^n$ for short. Further, when writing a probability mass function (PMF) we simply write $P_X(X=x)=p(x)$. An alphabet of any type is denoted by a calligraphic letter $\cal{X}$, and its size is denoted by $|{\cal{X}}|$.

In the rate distortion problem with feed-forward of delay $s=1$, as shown in Fig. \ref{RDprob}, we consider a general discrete, stationary, and ergodic source $\{X_n\}$, with the $n$th order probability distribution $p(x^n)$, alphabet $\cal{X}$ and reconstruction alphabet $\hat{\cal{X}}$. The normalized bounded distortion measure is defined as $d:{\cal{X}}^n\times\hat{\cal{X}}^n\to \mathbb{R}^+$ on pairs of sequences.
\begin{definition}[Code definition]
A $(n,2^{nR},D)$ source code with feed-forward of block length $n$ and rate $R$ consists of an encoder mapping $f$,
\begin{align}
f:&{\cal{X}}^n\mapsto\{1,2,...,2^{nR}\},\nonumber
\end{align}
and a sequence of decoder mappings $g_i,i=1,2,...,n$,
\begin{align}
g_i:&\{1,2,...,2^{nR}\}\times{\cal{X}}^{i-1}\mapsto\hat{\cal{X}},\ i=1,2,...,n.\label{MapEndDec}
\end{align}
The encoder maps a sequence $x^n$ to an index in $\{1,2,...,2^{nR}\}$. At time $i$, the decoder has the message that was sent and causal information of the source, $x^{i-1}$, and reconstructs the $i$th symbol sent, $\hat{x}_i$.
\end{definition}
\begin{definition}[Achievable rate]
A rate distortion with feed-forward pair $(R,D)$ is achievable if there exists a sequence of $(n,2^{nR},D)$-rate distortion codes with
\begin{align}
\lim_{n\to\infty}\ex{}{d(X^n,\hat{X}^n)}\leq D.\nonumber
\end{align}
\end{definition}
\begin{definition}[Rate distortion]
The rate distortion with feed-forward function $R(D)$ is the infimum of rates $R$ such that $(R,D)$ is achievable.
\end{definition}

In this paper, we define the mathematical expression for the rate distortion function as the following limit
\begin{align}
R^{(I)}(D)=\lim_{n\to\infty}R_n(D),\label{limoper}
\end{align}
where $R_n(D)$ is the $n$th order rate distortion function given by
\begin{align}
R_n(D)=\frac{1}{n}\min_{p(\hat{x}^n|x^n):\ex{}{d(X^n,\hat{X}^n)}\leq D}I(\hat{X}^n\to X^n).\nonumber
\end{align}
We show that the limit in (\ref{limoper}) exists, $R_n(D)$ is achievable and upper bounds $R^{(I)}(D)$ for all $n$. Further, we show that the rate distortion feed-forward function, $R(D)$, is equal to $R^{(I)}(D)$.
We also provide two ways to calculate numerically the value $R_n(D)$; using a BA-type algorithm and a geometric programming form.

We now state our main theorems.
\begin{theorem}[Achievability of $R_n(D)$]\label{ThRD}
For a discrete, stationary, ergodic source, and for any $D$, any $n$ and delay $s=1$, $R_n(D)$ is an achievable rate.
\end{theorem}
\begin{theorem}[Rate distortion feed-forward]\label{ThOperat}
For any distortion $D$, the operational rate distortion function $R(D)$ is equal to the mathematical expression, $R^{(I)}(D)$, where $R^{(I)}(D)$ is given by (\ref{limoper}).
\end{theorem}
\begin{theorem}\label{ThRdMaxgp}
The $n$th order rate distortion function $R_n(D)$ can be written in a geometric programming standard form as the following maximization problem
\begin{align}
R_n(D)=\max_{\lambda,\gamma(x^n),\{p'(x_i|x^{i-1},\hat{x}^{i})\}_{i=1}^n}\frac{1}{n}
    \left(-\lambda D+\sum_{x^n}p(x^n)\log\gamma(x^n)\right),\label{lowb1}
\end{align}
subject to the constraints:
\begin{align}
&\log(p(x^n))+\log(\gamma(x^n))-\lambda d(x^n,\hat{x}^n)-\sum_{i=1}^n \log{p'(x_i|x^{i-1},\hat{x}^i)}\leq 0,\ \ \forall\ x^n,\hat{x}^n,\nonumber\\
&\sum_{x_i}p'(x_i|x^{i-1},\hat{x}^i)=1,\ \ \forall\ i,\forall\ x^{i-1},\hat{x}^{i-1},\nonumber\\
&\lambda\geq0.\nonumber
\end{align}
\end{theorem}
\begin{theorem}[Algorithm for calculating $R_n(D)$]\label{Ths}
For a fixed source distribution $p(x^n)$, there exists an alternating minimization procedure in order to compute
\begin{align}
R_n(D)=\frac{1}{n}\min_{p(\hat{x}^n|x^n):\ex{}{d(X^n,\hat{X}^n)}\leq D}I(\hat{X}^n\rightarrow X^n).\label{RD1}
\end{align}
\end{theorem}

Proofs to Theorem \ref{ThRD} and \ref{ThOperat} are given in Section \ref{SecAchieve} and Section \ref{SecOperational}, respectively.
The proof for Theorem \ref{ThRdMaxgp} is in Section \ref{SecalgGp}, the algorithm in Theorem \ref{Ths} is described in Section \ref{secalgs} and proved in Section \ref{SecDer}.

\end{section}

\begin{section}{Achievability proof (Theorem \ref{ThRD}).}\label{SecAchieve}
In this section we show that if the source is stationary and ergodic, then $R_n(D)$ as given in (\ref{RD1}) is achievable for any $n$.
In order to do so, we first assume that the source is ergodic in blocks of length $n$, and show achievability. A source that is ergodic in blocks is one that, by looking at each $n$ letters as a single letter from a super alphabet, we obtain an ergodic super source (presented in \cite[Chapter 9.8]{Gal}). Then, for the general ergodic sources, we follow a claim given in \cite{Gal} about ergodic modes, as explained further on. The distortion  is assumed to be normalized, finite, and of the form
\begin{align}
d(x^n,\hat{x}^n)=\frac{1}{n}\sum_{i=1}^nd(x_{i-m}^i,\hat{x}_i),\label{distform}
\end{align}
for some $m$.
An example for such a distortion can be found in \cite{VenPra2} and in Section \ref{SecEx}, in an example called the stock-market.

\begin{theorem}\label{ThRDn}
Consider a discrete stationary source that is ergodic in blocks of length $n$. For any distortion $D$ such that $R_n(D)<\infty$ and $\delta>0$, and for any $L$ sufficiently large, there exists a codebook of trees $\cal{T}_C$ of length $L$ with $|{\cal{T}}_C|\leq2^{L(R_n(D)+\delta)}$ code trees for which the average distortion per letter satisfies $\ex{}{d(X^L,\hat{X}^L)}\leq D+\delta$.
\end{theorem}
\begin{proof}
Let $p(\hat{x}^n|x^n)$ be the transition probability that achieves the minimum $R_n(D)$ and let $p(\hat{x}^n||x^{n-1})$ be the causal conditioning probability that corresponds to $p(x^n)p(\hat{x}^n|x^n)$.
\begin{itemize}
\item   \underline{Code design.} For any $L$, consider the ensemble of codes $\cal{T}_C$ with $|{\cal{T}}_C|=\lfloor2^{L(R_n(D)+\delta)}\rfloor$ code trees of length $L$, where each code tree $\tau^L\in\cal{T}_C$ is a concatenation of $L/n$ sub-code trees of length $n$. Each sub-code tree is generated independently according to $p(\hat{x}^n||x^{n-1})$ as in Fig. \ref{codetree}.
\begin{figure}[h!]{
\psfrag{A}[][][1]{$p(\hat{x}_1)$} \psfrag{B}[][][1]{$p(\hat{x}_2|\hat{x}_1,x_1)$} \psfrag{C}[][][1]{$p(\hat{x}_3|\hat{x}^2_1,x^2_1)$} \psfrag{D}[][][1]{$p(\hat{x}_4)$} \psfrag{E}[][][1]{$p(\hat{x}_5|\hat{x}_4,x_4)$} \psfrag{F}[][][1]{$p(\hat{x}_6|\hat{x}^5_4,x^5_4)$}
\psfrag{I}[][][1]{$\hat{x}_1$} \psfrag{J}[][][1]{$\hat{x}_2$} \psfrag{K}[][][1]{$\hat{x}_2$}
\psfrag{L}[][][1]{$x_1=1$} \psfrag{M}[][][1]{$x_1=0$}
\psfrag{G}[][][1]{Code tree 1}\psfrag{H}[][][1]{Code tree 2}
 \centerline{ \includegraphics[scale=0.75]{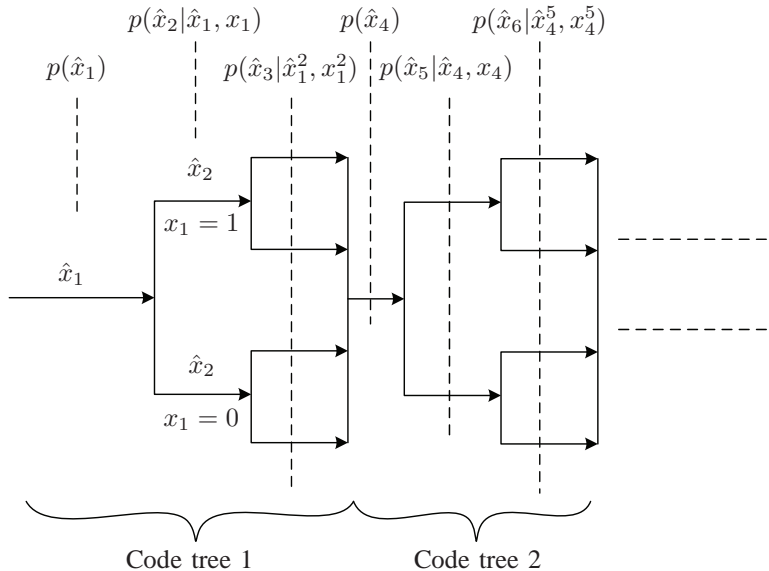}}
  \caption{Concatenation of two code trees, each of length $n=3$. The upper branches are for $x_i=1$, and the lower branches are for $x_i=0$.}
 \label{codetree}
}\end{figure}

\item   \underline{Encoder.} The encoder assigns a code tree $\tau^L\in\cal{T}_C$ for every $x^L$ such that $d(x^L,\hat{x}^L(\tau^L,x^{L-1}))$ is minimal. The sequence $\hat{x}^L(\tau^L,x^{L-1})$ is determined by walking on tree $\tau^L$, and following the branch $x^{L-1}$.
\item   \underline{Decoder.} At time $i$, the decoder possesses the index of the tree $\tau^L$ and causal information of the source $x^{i-1}$, and returns the symbol $\hat{x}_i(\tau^L,x^{i-1})$ that it produces.
\end{itemize}

Let us define a test channel as the conditional probability
\begin{align}
p_L(\hat{x}^L|x^L)&=\prod_{i=0}^{L/n-1} p(\hat{x}_{ni+1}^{ni+n}|x_{ni+1}^{ni+n}),\label{TestChan}
\end{align}
and the causal conditional probability
\begin{align}
p_L(\hat{x}^L||x^{L-1})&=\prod_{i=0}^{L/n-1}p(\hat{x}_{ni+1}^{ni+n}||x_{ni+1}^{ni+n-1}),\nonumber
\end{align}
where the distribution is according to
\begin{align}
P_{\hat{X}_{ni+1}^{ni+n}|X_{ni+1}^{ni+n}}(\hat{x}^n|x^n)&=P_{\hat{X}^n|X^n}(\hat{x}^n|x^n),\nonumber\\
P_{\hat{X}_{ni+1}^{ni+n}||X_{ni+1}^{ni+n-1}}(\hat{x}^n||x^{n-1})&=P_{\hat{X}^n||X^{n-1}}(\hat{x}^n||x^{n-1}).\nonumber
\end{align}
Moreover, we define for every code tree $\tau^L$ of length $L$ the measure
\begin{align}
I_n(\tau^L\rightarrow x^L)=\log\frac{p_L(\hat{x}^L|x^L)}{p_L(\hat{x}^L||x^{L-1})},
\end{align}
where $\hat{x}^L=\hat{x}^L(\tau^L,x^{L-1})$. Note that $I_n(\tau^L\to x^L)$ is not the directed information between the sequences $\hat{x}^L,\ x^L$, but simply a measure between a source sequence $x^L$ and the output, $\hat{x}^L$ of the test channel $p_L(\hat{x}^L|x^L)$, as defined in (\ref{TestChan}).

Let $\cal{T}$ be the set of all code trees of length $L$, and consider the following set,
\begin{align}
{\cal{A}}=\{\tau^L\in{\cal{T}},x^L\in{\cal{X}}^L:\ \textrm{either}\  I_n(\tau^L\rightarrow x^L)>L(R_n(D)+\delta/2)\ \ \ \textrm{or}\ \ \   d(x^L,\hat{x}^L(\tau^L,x^{L-1}))>L(D+\delta/2)\},\label{TestA}
\end{align}
and let $p_t({\cal{A}})$ be the probability of the set ${\cal{A}}$ on the test channel ensemble.

Let us use the notation
\begin{align}
\hat{x}^L({\cal{T}}_C,x^{L-1})=\hat{x}^L\left(\arg\min_{\tau^L\in{\cal{T}}_C}d\big(x^L,\hat{x}^L(\tau^L,x^{L-1})\big),x^L\right),\nonumber
\end{align}
where ${\cal{T}}_C$ is the ensemble of code trees as described in the coding scheme.
Now, let $p_c(d(X^L,\hat{x}^L({\cal{T}}_C,X^{L-1}))>LD)$ be the probability over the ensemble of codes ${\cal{T}}_C$ and source sequences such that the distortion exceeds $LD$. We wish to give an upper bound to the probability $p_c(d(X^L,\hat{x}^L({\cal{T}}_C,X^{L-1}))>LD)$; for this we use the following lemma.
\begin{lemma}
For a given source $\{X_i\}_{i\geq1}$ and test channel, we have the following inequality
\begin{align}
p_c\left(d(X^L,\hat{x}^L({\cal{T}}_C,X^{L-1}))>LD\right)\leq p_t({\cal{A}})+\exp\{-|{\cal{T}}_C|2^{-LR_n(D)}\},
\end{align}
where the set $A$ is described in (\ref{TestA}).
\end{lemma}
\textit{Proof.} We first write $p_c\left(d(X^L,\hat{x}^L({\cal{T}}_C,X^{L-1}))>LD\right)$ as
\begin{align}
p_c\left(d(X^L,\hat{x}^L({\cal{T}}_C,X^{L-1}))>LD\right)=\sum_{x^L\in{\cal{X}}^L}p(x^L)p_c\left(d(X^L,\hat{x}^L({\cal{T}}_C,X^{L-1}))>LD|X^L=x^L\right).\nonumber
\end{align}
For every $x^L$, let us define the set ${\cal{A}}_{x^L}$ as the set of all code trees $\tau^L\in{\cal{T}}$ for which $(\tau^L,x^L)\in {\cal{A}}$,
\begin{align}
{\cal{A}}_{x^L}=\{\tau^L\in{\cal{T}}:\ \textrm{either}\  I_n(\tau^L\rightarrow x^L)>L(R_n(D)+\delta/2)\ \ \ \textrm{or}\ \ \   d(x^L,\hat{x}^L(\tau^L,x^{L-1}))>L(D+\delta/2)\}.\label{TestAxL}
\end{align}

We observe that $d(x^L,\hat{x}^L({\cal{T}}_C,x^{L-1}))>LD$ for a given $x^L$ only if $d(x^L,\hat{x}^L(\tau^L,x^{L-1}))>LD$ for every $\tau^L\in{\cal{T}}_C$. Thus, $d(x^L,\hat{x}^L({\cal{T}}_C,x^{L-1}))>LD$ only if $\tau^L\in {\cal{A}}_{x^L}$ for every $\tau^L\in {\cal{T}}_C$. Since $\tau^L$ are independently chosen,
\begin{align}
p_c\left(d(X^L,\hat{x}^L({\cal{T}}_C,X^{L-1}))>LD|X^L=x^L\right)&\leq \left(p_t({\cal{A}}_{x^L})\right)^{|{\cal{T}}_C|}\nonumber\\
&=\left(1-p_t({\cal{A}}_{x^L}^c)\right)^{|{\cal{T}}_C|},\nonumber
\end{align}
where ${\cal{A}}_{x^L}^c$ is the complement set of ${\cal{A}}_{x^L}$. We note that the probability that tree $\tau^L$ being in ${\cal{A}}_{x^L}^c$ depends only on the branch associated with $x^{L}$. In other words, if a tree $\tau^L\in {\cal{A}}_{x^L}^c$, then all other trees with the same branch associated with $x^L$ is in ${\cal{A}}_{x^L}^c$ as well; the same goes for ${\cal{A}}_{x^L}$. Hence, we can divide the set of all code trees ${\cal{T}}$ into disjoint subsets $B_{x^L,\hat{x}^L}$ that have the same branch associated with $x^{L-1}$, i.e.,
\begin{align}
B_{x^L,\hat{x}^L}=\{\tau^L\in{\cal{T}}: \tau^L(x^{L-1})=\hat{x}^{L}\}\nonumber,
\end{align}
where $\tau^L(x^{L-1})$ is a walk on tree $\tau^L$ over the branch $x^{L-1}$. Clearly, the probability of each subset $B_{x^L,\hat{x}^L}$ is
\begin{align}
p_t(B_{x^L,\hat{x}^L})=p_L(\hat{x}^L||x^{L-1})\nonumber
\end{align}
since the left hand side is a summation of the probabilities of all trees with the same branch associated with $x^L$, and we are left with the probability of that one branch.

Now, for every $\tau^L\in B_{x^L,\hat{x}^L}\subset {\cal{A}}_{x^L}^c$, and due to the definition of ${\cal{A}}_{x^L}^c$, we have
\begin{align}
I_n(\tau^L\rightarrow x^L)=\log\frac{p_L(\hat{x}^L|x^L)}{p_L(\hat{x}^L||x^{L-1})}\leq LR_n(D).\nonumber
\end{align}
Therefore,
\begin{align}
p_L(\hat{x}^L||x^{L-1})\geq p_L(\hat{x}^L|x^L)2^{-LR_n(D)},\label{IneqAch}
\end{align}
and we obtain that
\begin{align}
p_c\left(d(X^L,\hat{x}^L({\cal{T}}_C,X^{L-1}))>LD|X^L=x^L\right)&\leq\left(1-p_t({\cal{A}}_{x^L}^c)\right)^{|{\cal{T}}_C|}\nonumber\\
&=\left(1-\sum_{B_{x^L,\hat{x}^L}\subset {\cal{A}}_{x^L}^c}p_t(B_{x^L,\hat{x}^L})\right)^{|{\cal{T}}_C|}\nonumber\\
&=\left(1-\sum_{\hat{x^L}:B_{x^L,\hat{x}^L}\subset {\cal{A}}_{x^L}^c}p_L(\hat{x}^L||x^{L-1})\right)^{|{\cal{T}}_C|}\nonumber\\
&\stackrel{(a)}{\leq}\left(1-2^{-LR_n(D)}\sum_{\hat{x}^L:B_{x^L,\hat{x}^L}\subset {\cal{A}}_{x^L}^c}p_L(\hat{x}^L|x^L)\right)^{|{\cal{T}}_C|},\nonumber
\end{align}
where (a) follows the inequality in equation (\ref{IneqAch}).

Using the inequality $(1-ab)^k\leq 1-a+\exp\{-bk\}$, and taking $a=\sum_{\hat{x}^L:B_{x^L,\hat{x}^L}\subset {\cal{A}}_{x^L}^c}p_L(\hat{x}^L|x^L)$, $b=2^{-LR_n(D)}$, we find
\begin{align}
p_c\left(d(X^L,\hat{x}^L({\cal{T}}_C,X^{L-1}))>LD|X^L=x^L\right)&\leq1-\sum_{\hat{x}^L:B_{x^L,\hat{x}^L}\subset {\cal{A}}_{x^L}^c}
    p_L(\hat{x}^L|x^L)+\exp\{-|{\cal{T}}_C|2^{-LR_n(D)}\}.\nonumber
\end{align}
By taking a sum over $x^L$ we remain with
\begin{align}
p_c\left(d(X^L,\hat{x}^L({\cal{T}}_C,X^{L-1}))>LD\right)&=\sum_{x^L}p(x^L)p_c\left(d(X^L,\hat{x}^L({\cal{T}}_C,X^{L-1}))>LD|X^L=x^L\right)\nonumber\\
&\leq\sum_{x^L}p(x^L)\left(1-\sum_{\hat{x}^L:B_{x^L,\hat{x}^L}\subset
    {\cal{A}}_{x^L}^c}p_L(\hat{x}^L|x^L)+\exp\{-|{\cal{T}}_C|2^{-LR_n(D)}\}\right)\nonumber\\
&=1-\sum_{x^L}\sum_{\hat{x}^L:B_{x^L,\hat{x}^L}\subset {\cal{A}}_{x^L}^c}p(x^L,\hat{x}^L)+\exp\{-|{\cal{T}}_C|2^{-LR_n(D)}\}.\label{sumprob}
\end{align}
Note, that
\begin{align}
\sum_{x^L}\sum_{\hat{x}^L:B_{x^L,\hat{x}^L}\subset {\cal{A}}_{x^L}^c}p(x^L,\hat{x}^L)&=
    \sum_{x^L}\sum_{\hat{x}^L:B_{x^L,\hat{x}^L}\subset {\cal{A}}_{x^L}^c}\sum_{\tau^L\in{\cal{T}}}p(x^L,\hat{x}^L,\tau^L)\nonumber\\
&\geq\sum_{x^L}\sum_{\hat{x}^L:B_{x^L,\hat{x}^L}\subset {\cal{A}}_{x^L}^c}\sum_{\tau^L\in B_{x^L,\hat{x}^L}}p(x^L,\hat{x}^L,\tau^L)\nonumber\\
&\stackrel{(a)}{=}\sum_{x^L}\sum_{B_{x^L,\hat{x}^L}\subset {\cal{A}}_{x^L}^c}\sum_{\tau^L\in B_{x^L,\hat{x}^L}}p(x^L,\tau^L)\nonumber\\
&=\sum_{x^L}\sum_{\tau^L\in {\cal{A}}_{x^L}^c}p(x^L,\tau^L)\nonumber\\
&=p_t({\cal{A}}^c),\nonumber
\end{align}
where (a) follows the fact that if $\tau^L\in B_{x^L,\hat{x}^L}$, then $\hat{x}^L$ is determined by the tree $\tau^L$ and the branch $x^L$.
Now, continuing from equation (\ref{sumprob}), we obtain
\begin{align}
p_c\left(d(X^L,\hat{x}^L({\cal{T}}_C,X^{L-1}))>LD\right)&\leq1-p_t({\cal{A}}^c)+\exp\{-|{\cal{T}}_C|2^{-LR_n(D)}\}\nonumber\\
&=p_t({\cal{A}})+\exp\{-|{\cal{T}}_C|2^{-LR_n(D)}\}.\label{per}
\end{align}\hfill\QED

We now use the result in (\ref{per}) in order to complete the proof of the theorem. Furthermore, we can see that the average distortion of the code satisfies
\begin{align}
\ex{}{d(X^L,\hat{X}^L}(\leq (D+\delta/2)+p_c\left(d(X^L,\hat{x}^L({\cal{T}}_C,X^{L-1}))>L(D+\delta/2)\right)\cdot\sup_{x^L,\hat{x}^L}{d(x^L,\hat{x}^L)}.\nonumber
\end{align}
This arises, as in \cite[Th. 9.3.1]{Gal}, from upper bounding the distortion by $D+\delta/2$ when the $d(x^L,\hat{x}^L)\leq D+\delta/2$, and by
\begin{align}
\sup_{x^L,\hat{x}^L}{d(x^L,\hat{x}^L)}\nonumber
\end{align}
otherwise. By choosing $|{\cal{T}}_C|=\lfloor2^{L(R_n(D)+\delta)}\rfloor$, the last term in (\ref{per}) goes to zero with increasing $L$. Furthermore, the first term is bounded by
\begin{align}
p_t({\cal{A}})\leq p_t\{x^L\in{\cal{X}}^L&,\tau^L\in{\cal{T}}:\ I_n(\tau^L\rightarrow x^L)>L(R_n(D)+\delta/2)\}\nonumber\\
&+p_t\{x^L\in{\cal{X}}^L,\tau^L\in{\cal{T}}:\ d(x^L,\hat{x}^L(\tau^L,x^{L-1}))>L(D+\delta/2)\}.\label{pAbound}
\end{align}

Note that
\begin{align}
p_t\left(I_n(\tau^L\rightarrow x^L)>L(\frac{1}{n}R_n(D)+\delta/2)\right) =p_t\left(\frac{1}{L}\sum_{i=1}^{L/n-1}\log\frac{p(\hat{x}_{ni+1}^{ni+n}|x_{ni+1}^{ni+n})}{p(\hat{x}_{ni+1}^{ni+n}||x_{ni+1}^{ni+n-1})}
    >R_n(D)+\delta/2\right).\nonumber
\end{align}
As assumed, the source is ergodic in blocks of length $n$. Furthermore, the test channel is defined to be memoryless for blocks of length $n$, and hence the joint process is ergodic in blocks of length $n$. Thus, with probability 1,
\begin{align}
\frac{1}{n}\lim_{L\rightarrow\infty}\frac{1}{L/n}\sum_{i=0}^{L/n-1}\log\frac{p(\hat{x}_{ni+1}^{ni+n}|x_{ni+1}^{ni+n})}{p(\hat{x}_{ni+1}^{ni+n}||x_{ni+1}^{ni+n-1})}&=
    \frac{1}{n}\ex{}{\log\frac{p(\hat{x}^n|x^n)}{p(\hat{x}^n||x^{n-1})}}\nonumber\\
&=R_n(D).\nonumber
\end{align}
Therefore, the probability of the first term in (\ref{pAbound}) goes to zero as $L$ goes to infinity, and the same goes to the second term due to the definition of the distortion.
In order to finish the proof, and due to the fact that $p_c$ goes to zero with increasing $L$ and the fact that the distortion is finite, we can choose $L$ large enough such that
\begin{align}
p_c\left(d(X^L,\hat{x}^L({\cal{T}}_C,X^{L-1}))>L(D+\delta/2)\right)\cdot\sup_{x^L,\hat{x}^L}{d(x^L,\hat{x}^L)}\leq\delta/2.\nonumber
\end{align}
In this case, we obtain $D_L\leq D+\delta$, and hence the rate $R_n(D)$ is achievable for sources that are ergodic in blocks of length $n$.
\end{proof}

Much like in Gallager's proof for the case where there is no feed-forward, we note that not all ergodic sources are also ergodic in blocks, and we need to address these cases as well. For that purpose, we need \cite[Lemma 9.8.2]{Gal} for ergodic sources. We recall, that a discrete stationary source is ergodic if and only if every invariant set of sequences under a shift operator $T$ is of probability 1 or 0. In \cite[Chapter 9.8]{Gal}, the author looks at the operator $T^n$, i.e., a shift of $n$ places, and considers an invariant set $S_0$, $p(S_0)>0$, with respect to $T^n$. In Lemma 9.8.2 in \cite{Gal}, it is stated that one can separate the source $S$ to $n'$ invariant subsets $\{S_i=T^i(S_0)\}_{i=0}^{n'-1}$, $p(S_i)=\frac{1}{n'}$, with regard to $T^n$, such that $n'$ divides $n$ and the sets $S_i,\ S_j$ are disjoint except, perhaps, an intersection of zero probability. These subsets are called \textit{ergodic modes}, due to the fact that each invariant subset of them under the operator $T^n$ is of probability 0 or $\frac{1}{n'}$. In other words, conditional on an ergodic mode $S_i$ each invariant subset of it with respect to $T^n$, is of probability 0 or 1.

Recall, that by definition,
\begin{align}
R_n(D)=\frac{1}{n}I_n(\hat{X}^n\rightarrow X^n),\nonumber
\end{align}
where the right-hand side is the average directed information between the source and reconstruction, determined according to $p(\hat{x}^n|x^n)$ that achieves $R_n(D)$. Let $I_n(\hat{X}^n\rightarrow X^n|i)$ be the average directed information between a source sequence from the $i$th ergodic mode and the ensemble of codes, using the probability $p(\hat{x}^n|x^n)$ which achieves $R_n(D)$. Note that the directed information can be written as
\begin{align}
I_n(\hat{X}^n\rightarrow X^n)&=\sum_{x^n,\hat{x}^n}p(x^n)p(\hat{x}^n|x^n)\log\frac{p(\hat{x}^n|x^n)}{p(\hat{x}^n||x^{n-1})}\nonumber\\
&=\sum_{x^n,\hat{x}^n}p(x^n)p(\hat{x}^n|x^n)\log\frac{p(\hat{x}^n|x^n)p(x^n)}{p(\hat{x}^n||x^{n-1})p(x^n)}\nonumber\\
&=D\left(p(x^n)p(\hat{x}^n|x^n)||p(\hat{x}^n||x^{n-1})p(x^n)\right),\nonumber
\end{align}
which is convex over the input probability $p(x^n)$. Thus,
\begin{align}
I_n(\hat{X}^n\rightarrow X^n)\geq\frac{1}{n'}\sum_{i=0}^{n'-1}I_n(\hat{X}^n\rightarrow X^n|i).\label{supsourceineq}
\end{align}
We observe that $\frac{1}{n}I_n(\hat{X}^n\rightarrow X^n|i)$ is an upper bound to the $n$th order rate distortion function conditional on the $i$th ergodic mode. From Theorem \ref{ThRDn}, we know that there exists a codebook ${\cal{T}}_{C_i}$ with $|{\cal{T}}_{C_i}|=\lfloor2^{L(\frac{1}{n}I_n(\hat{X}^n\rightarrow X^n|i)+\delta)}\rfloor$ code trees of length $L$ such that the average distortion constraint holds. Another observation is that if a codebook ${\cal{T}}_{C_i}$ satisfies the distortion constraint, conditional on the ergodic mode $S_i$, then it has the same effect conditional on the ergodic mode $T(S_{i-1})$. In other words, we can encode not only a source sequence from $S_{i-i}$ with ${\cal{T}}_{C_{i-1}}$, but also a shift of the a source sequence in $S_{i-1}$ with ${\cal{T}}_{C_{i}}$. We use these observations while constructing the codebook.

We can now prove Theorem \ref{ThRD}, i.e., the achievability of $R_n(D)$, where the source is ergodic and stationary. An equivalent version of Theorem \ref{ThRD} is the following:
let $R_n(D)$ be the $n$th order rate distortion function for a discrete, stationary, and ergodic source. For any $D$ such that $R_n(D)<\infty$, and $\delta>0$, and any $L$ sufficiently large, there exists a codebook of trees $\cal{T}_C$ of length $L$ with $|{\cal{T}}_C|\leq2^{L(R_n(D)+\delta)}$ code trees for which the average distortion per letter satisfies $\ex{}{d(X^n,\hat{X}^n}\leq D+\delta$.
\begin{proof}[Proof of Theorem \ref{ThRD}]
Let $p(\hat{x}^n|x^n)$ be the transition probability that achieves $R_n(D)$ and let $p(\hat{x}^n||x^{n-1})$ be the causal conditioning probability that corresponds to $p(x^n)p(\hat{x}^n|x^n)$.
\begin{itemize}
\item   \underline{Code design.} For any $L$ and any ergodic mode $S_i$, $0\leq i\leq n'$, construct an ensemble of codes ${\cal{T}}_{C_i}$, with $|{\cal{T}}_{C_i}|=\lfloor2^{L(\frac{1}{n}I_n(\hat{X}^n\rightarrow X^n|i)+\delta)}\rfloor$ 'little' code trees of length $L$, where each 'little' code tree is generated according to $p(\hat{x}^L||x^{L-1})$, as in Fig. \ref{codetree} in Theorem \ref{ThRDn} above. Now, for every $0\leq i\leq n'$, the $i$th codebook is an ensemble of 'big' code trees, which are concatenation of $n'$ 'little' code trees, starting from one in ${\cal{T}}_{C_i}$, and followed by one from ${\cal{T}}_{C_{i+1}}$ to one from ${\cal{T}}_{C_{n'+i-1}}$, where the index is calculated modiolus $n'$. In the example of a 'big' code tree in Fig. \ref{ncodetree} we see additional letters at the end of each 'little' code tree, i.e., in positions $L+1,\ 2(L+1),...,n'(L+1)$, that are fixed. The purpose of the fixed letters is to shift the sequence and encode it with a codetree from the sequential codebook. Note, that the overall length of a code tree sums up to $L'=Ln'+n'$.
    \begin{figure}[h!]{
    \psfrag{A}[][][1]{Codetree from ${\cal{T}}_{C_i}$} \psfrag{B}[][][1]{Codetree from ${\cal{T}}_{C_{i+1}}$}
    \psfrag{C}[][][1]{Codetree from ${\cal{T}}_{C_{i+2}}$} \psfrag{D}[][][1]{Fixed letters}
    \psfrag{G}[][][1]{Code tree 1}\psfrag{H}[][][1]{Code tree 2}
     \centerline{ \includegraphics[scale=0.75]{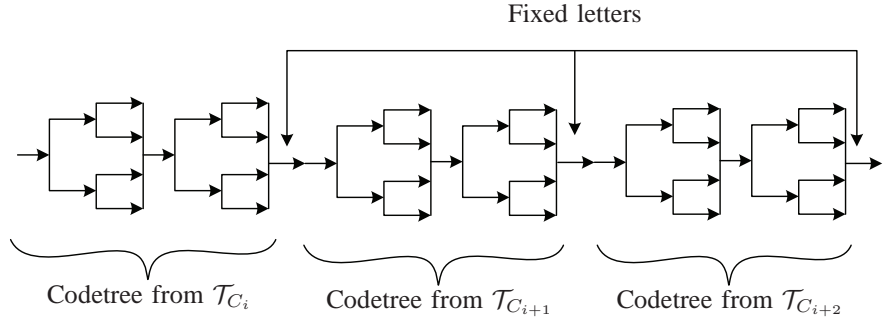}}
      \caption{A code tree from the $i$th codebook, $n=n'=3$, $L=6$.}
     \label{ncodetree}
    }\end{figure}

\item   \underline{Encoder.} For every $i$, the encoder assigns for every source sequence $x^{L'}\in S_i$ a code tree $\tau^{L'}$ from the $i$th codebook,  such that $d(x^{L'},\hat{x}^{L'}(\tau^{L'},x^{L'-1}))$ is minimal. The sequence $\hat{x}^{L'}(\tau^{L'},x^{L'-1})$ is determined by walking on tree $\tau^{L'}$, and following the branch $x^{L'-1}$.
\item   \underline{Decoder.} The decoder receives a tree $\tau^{L'}$ and causal information of $x^{L'}$ and returns the sequence $\hat{x}^{L'}$ that it produces.
\end{itemize}

Since the distortion constraint for every ergodic mode is satisfied due to Theorem \ref{ThRDn}, the overall distortion is satisfied as well. The additional fixed letters are of unknown distortion, but due to the face that the distortion is bounded, their contribution is negligible for large $L$. Moreover, note that for every $i$, the $i$th codebook is of the same size. Thus, the overall size of the codebook is
\begin{align}
|{\cal{T}}_C|&=n'\prod_{i=0}^{n'-1}|{\cal{T}}_{C_i}|\nonumber\\
&\leq n'\prod_{i=0}^{n'-1}2^{L(\frac{1}{n}I_n(\hat{X}^n\rightarrow X^n|i)+\delta)}\nonumber\\
&=2^{L(\frac{1}{n}\sum_{i=0}^{n'-1}I_n(\hat{X}^n\rightarrow X^n|i)+n'\delta+\frac{\log(n')}{L})}\nonumber\\
&\leq2^{L(\frac{n'}{n}I_n(\hat{X}^n\rightarrow X^n)+n'\delta+\frac{\log(n')}{L})}\nonumber\\
&=2^{Ln'(R_n(D)+\delta+\frac{\log(n')}{Ln'})}\nonumber\\
&\leq2^{(Ln'+n')(R_n(D)+\delta+\frac{\log(n')}{Ln'})}.\nonumber
\end{align}
Recall that $L'=Ln'+n'$, and by letting $\delta'=\delta+\frac{\log(n')}{Ln'}$ we conclude that $R_n(D)$ is an achievable rate for the general ergodic source, as required.
\end{proof}
\end{section}

\begin{section}{Proof that $R(D)=R^{(I)}(D)$ (Theorem \ref{ThOperat}).}\label{SecOperational}
In this section we show that the operational description of the rate distortion with feed-forward is equal to the mathematical one given in (\ref{RDI}). This will be done first by showing that the mathematical expression $R^{(I)}(D)$ is achievable, and then by showing that it is a lower bound to the rate distortion function.
We recall that
\begin{align}
R^{(I)}(D)=\lim_{n\to\infty}\frac{1}{n}\min_{p(\hat{x}^n|x^n):\ex{}{d(X^n,\hat{X}^n)}\leq D}I(\hat{X}^n\to X^n).\label{RDI}
\end{align}

To show that $R^{(I)}(D)$ is achievable we first need to show that the limit of the sequence $\{R_n(D)\}$ exists. For this purpose, we use the following lemma.
\begin{lemma}\label{LemSubAdd}
The sequence $R_n(D)$,
\begin{align}
R_n(D)=\frac{1}{n}\min_{p(\hat{x}^n|x^n):\ex{}{d(X^n,\hat{X}^n)}\leq D}I(\hat{X}^n\rightarrow X^n),\nonumber
\end{align}
is a sub-additive sequence, and thus
\begin{align}
\inf_n R_n(D)=\lim_{n\rightarrow\infty}R_n(D).\nonumber
\end{align}
\end{lemma}
Note, that a sequence $\{a_n\}$ is called sub-additive if for all $m,l$,
\begin{align}
(m+l)a_{m+l}\leq ma_m+la_l.\nonumber
\end{align}
The proof for Lemma \ref{LemSubAdd} is given in App. \ref{Appsubadd}.

We now state a lemma for the achievability of $R^{(I)}(D)$.
\begin{lemma}[Achievability of $R^{(I)}(D)$] \label{LemAch}
The mathematical expression for the rate distortion feed-forward $R^{(I)}(D)$ is achievable, and thus upper bounds $R(D)$.
\end{lemma}
\begin{proof}
We showed in Theorem \ref{ThRD} that for any $n$, $R_n(D)$ is achievable. Further, in Lemma \ref{LemSubAdd} we show that the limit exists and equal to the infimum, and hence is achievable too. Therefore, we conclude that the mathematical expression $R^{(I)}(D)$ is achievable, and forms an upper bound to the operational description $R(D)$.
\end{proof}

To show that $R^{(I)}(D)$ is a lower bound to the rate distortion function, we provide the following lemma
\begin{lemma}[Converse] \label{LemConverse}
the mathematical expression $R^{(I)}(D)$ is a lower bound to the operational rate distortion function.
\end{lemma}
For the completeness of the paper, we provide the proof of Lemma \ref{LemConverse}, this in App. \ref{Appconv}. However, similar proof was presented by Venkataramana and Pradhan in \cite{VenPra}, and their expressions involved limit in probability of the entropy and directed information as described in Section \ref{SecIntro}.

\begin{proof}[Proof of Theorem \ref{ThOperat}]
Combining Lemmas \ref{LemAch}, \ref{LemConverse} provides us with the proof for our fundamental theorem, stated in Section \ref{SecProbRes}, i.e., the operational rate distortion function $R(D)$ is equal to the mathematical one, $R^{(I)}(D)$.
\end{proof}
\end{section}

\begin{section}{Geometric programming form to $R_n(D)$ (Theorem \ref{ThRdMaxgp})}\label{SecalgGp}
In this section we show that the $n$th order rate distortion function with feed-forward $R_n(D)$ can be given as a maximization problem, written in a standard form of geometric programming. For this purpose we first state the following theorem.
\begin{theorem}\label{ThRdMax}
The $n$th order rate distortion function, $R_n(D)$, can be written as the following maximization problem
\begin{align}
R_n(D)=\max_{\lambda\geq0,\gamma(x^n)}\frac{1}{n}
    \left(-\lambda D+\sum_{x^n}p(x^n)\log\gamma(x^n)\right),\label{lowb4}
\end{align}
where, for some causal conditioned probability $p'(x^n||\hat{x}^n)$, $\gamma(x^n)$ satisfies the inequality constraint
\begin{align}
p(x^n)\gamma(x^n)2^{-\lambda d(x^n,\hat{x}^n)}\leq p'(x^n||\hat{x}^n).\label{cons3}
\end{align}
\end{theorem}
In App. \ref{Apprdmax} we provide two proofs for Theorem \ref{ThRdMax}; the first is similar to Berger's proof in \cite{Berger} for the regular rate distortion function based on the inequality $\log(y)\geq1-\frac{1}{y}$, and the second uses the Lagrange duality as presented in \cite{Boyd} and \cite{ChiBoyd} that transforms a minimization problem to a maximization one.. App. \ref{Apprdmax} also includes the connection between the rate distortion function and the parameter $\lambda$, which states that the slope of $R_n(D)$ in point $D$ is $-\frac{\lambda}{n}$.

\begin{proof}[Proof of Theorem \ref{ThRdMaxgp}]
Considering the theorem above, our interest now is to adjust the constraints in order to obtain a geometric programming form. We note that the optimization problem in (\ref{lowb4}) does not change if we maximize over $p'(x^n||\hat{x}^n)$ as well, and the constraint (\ref{cons3}) is no longer for some $p'$, i.e.,
\begin{align}
R_n(D)=\max_{\lambda\geq0,\gamma(x^n),p'(x^n||\hat{x}^n)}\frac{1}{n}
    \left(-\lambda D+\sum_{x^n}p(x^n)\log\gamma(x^n)\right),\label{lowb5}
\end{align}
where $\gamma(x^n),\ p'(x^n||\hat{x}^n)$ satisfy the inequality constraint
\begin{align}
p(x^n)\gamma(x^n)2^{-\lambda d(x^n,\hat{x}^n)}\leq p'(x^n||\hat{x}^n).\label{cons4}
\end{align}
The above statement is true since, on the one hand, the maximization in (\ref{lowb4}) increases upon maximizing over another variable,  $p'(x^n||\hat{x}^n)$, as in (\ref{lowb5}); on the other hand, the variable $\gamma^*(x^n),\ p'^*(x^n||\hat{x}^n)$ that achieves (\ref{lowb5}) satisfy the constraint (\ref{cons3}) in Theorem \ref{ThRdMax}, and hence the maximization problem in (\ref{lowb5}) cannot be greater than the one in (\ref{lowb4}).

To obtain a geometric programming standard form we transform the constraint in (\ref{cons4}), such that
\begin{align}
p(x^n)\gamma(x^n)2^{-\lambda d(x^n,\hat{x}^n)}p'(x^n||\hat{x}^n)^{-1}\leq1.\nonumber
\end{align}
Taking the $\log$ of both sides, we obtain
\begin{align}
&\log(p(x^n))+\log(\gamma(x^n))-\lambda d(x^n,\hat{x}^n)-\sum_{i=1}^n \log{p'(x^n||\hat{x}^n)}\leq 0.\nonumber
\end{align}

Note that maximizing over $p'(x^n||\hat{x}^n)$ is the same as maximizing over its products $\{p'(x_i|x^{i-1},\hat{x}^i)\}_{i=1}^n$\cite[Lemma 3]{PermuterWeissmanGoldsmith}.
Therefore, we can conclude that the rate distortion with feed-forward $R_n(D)$ can be given as a geometric programming maximization form,
\begin{align}
R_n(D)=\max_{\lambda,\gamma(x^n),\{p'(x_i|x^{i-1},\hat{x}^i)\}_{i=1}^n}\frac{1}{n}
    \left(-\lambda D+\sum_{x^n}p(x^n)\log\gamma(x^n)\right),\nonumber
\end{align}
subject to
\begin{align}
&\log(p(x^n))+\log(\gamma(x^n))-\lambda d(x^n,\hat{x}^n)-\sum_{i=1}^n \log{p'(x_i|x^{i-1},\hat{x}^i)}\leq 0,\ \ \forall\ x^n,\hat{x}^n,\nonumber\\
&\sum_{x_i}p'(x_i|x^{i-1},\hat{x}^i)=1,\ \ \forall\ i,\forall\ x^{i-1},\hat{x}^{i-1},\nonumber\\
&\lambda\geq0.\nonumber
\end{align}
Hence, we obtain a standard form of geometrical programming. This GP problem can be solved using standard convex optimization tools.
\end{proof}

\end{section}

\begin{section}{Extension of the BAA for rate distortion with feed-forward}\label{secalgs}
In this section we describe an algorithm for calculating $R_n(D)$, where
\begin{align}
R_n(D)=\frac{1}{n}\min_{r(\hat{x}^n|x^n):\ex{}{d(X^n,\hat{X}^n)}\leq D}I(\hat{X}^n\rightarrow X^n),\label{Rlim3}
\end{align}
using the alternating minimization procedure. This method was first used by Blahut and Arimoto \cite{Bla}, \cite{Ari} to obtain a numerical solution for the i.i.d. source rate distortion and for the memoryless channel capacity. Recently, in \cite{NaissPermuter1} we extended this method for finding the global maximum of the following optimization problem-
\begin{align}
C_n=\frac{1}{n}\max_{p(x^n||y^{n-1})}I(X^n\rightarrow Y^n),\nonumber
\end{align}
and we apply similar methods here.

Before we describe the algorithm, let us denote by $r=r(\hat{x}^n|x^n),\ q=q(\hat{x}^n||x^{n-1})$ the PMFs that are participating in the minimization. Further, let us consider the double optimization problem given by
\begin{align}
R_n(D)=\frac{1}{n}\left[-\lambda D+\min_{r,q}K(r,q)\right],\label{Rlimdouble}
\end{align}
where
\begin{align}
K(r,q)=I_{FF}(r,q)+\lambda\ex{r}{d(X^n,\hat{X}^n)},\nonumber
\end{align}
and $I_{FF}(r,q)$ is the directed information that can be written as
\begin{align}
I_{FF}(r,q)=I(\hat{X}^n\rightarrow X^n)=\sum_{\hat{x}^n,x^n}{p(x^n)r(\hat{x}^n|x^n)\log{\frac{r(\hat{x}^n|x^n)}{q(\hat{x}^n||x^{n-1})}}}.\label{directed4}
\end{align}
In Section \ref{SecDer} we show that the double optimization problem given in (\ref{Rlimdouble}) is equal to the one given in (\ref{Rlim3}).
Equations (\ref{Rlimdouble}), (\ref{directed4}) allow us to apply the alternating minimization procedure.

\begin{subsection}{Description of the algorithm}
In Algorithm \ref{algs} we present the steps required to minimize the directed information where the input PMF $p(x^n)$ is fixed.
\begin{algorithm}
\caption{Iterative algorithm for calculating $R_n(D)$, where $p(x^n)$ is fixed.}\label{algs}
\begin{itemize}
\item[(a)] Fix a value of $\lambda\geq 0$ that determines a point on the $R_n(D)$ curve.\\
\item[(b)] Start from a random causally conditioned point $q^0(\hat{x}^n||x^{n-1})$. Usually we start from a uniform distribution, i.e., $q^0(\hat{x}^n||x^{n-1})=2^{-n}$ for every $(x^n,\hat{x}^n)$.\\
\item[(c)] Set $k=1$. \\
\item[(d)] Compute $r^{k}(\hat{x}^n|x^n)$ using the formula
\begin{equation}
r^{k}(\hat{x}^n|x^n)=\frac{q^{k-1}(\hat{x}^n||x^{n-1})2^{-\lambda d(x^n,\hat{x}^n)}}{\sum_{\hat{x}^n}{q^{k-1}(\hat{x}^n||x^{n-1})2^{-\lambda d(x^n,\hat{x}^n)}}}.\nonumber
\end{equation}
\item[(e)] Calculate the joint probability $p(x^n,\hat{x}^n)=p(x^n)r^{k}(\hat{x}^n|x^n)$, and deduce the causal conditioned PMF $q^{k}(\hat{x}^n||x^{n-1})$ as in (\ref{CausalPr}).\\
\item[(f)] Calculate the parameter
\begin{align}
c^k_{\hat{x}^n,x^{n-1}}=\frac{q^{k}(\hat{x}^n||x^{n-1})}{q^{k-1}(\hat{x}^n||x^{n-1})}.\nonumber
\end{align}
\item[(g)] Calculate
\begin{align}
F=\log{\max_{\hat{x}^n,x^{n-1}}c^k_{\hat{x}^n,x^{n-1}}}-\sum_{x^n,\hat{x}^n}p(x^n)r^k(\hat{x}^n|x^n)\log{c^k_{\hat{x}^n,x^{n-1}}}.\nonumber
\end{align}
\item[(h)] If $F\geq\epsilon$, set $k:=k+1$, and return to (d).\\
\item[(i)] The rate distortion function, with distortion $D_k=\sum_{\hat{x}^n,x^n}p(x^n)r^k(\hat{x}^n|x^n)d(x^n,\hat{x}^n)$, is
\begin{align}
R^k_n(D_k)=\frac{1}{n}\sum_{x^n,\hat{x}^n}{p(x^n)r^k(\hat{x}^n|x^n)\log{\frac{r^k(\hat{x}^n|x^n)}{q^k(\hat{x}^n||x^{n-1})}}}.\nonumber
\end{align}
\end{itemize}
\end{algorithm}
The parameter $\lambda$ is used in the Lagrangian with which we optimize the directed information. The value of $D_k$ and hence $R_n(D_k)$ depends on $\lambda$; thus choosing $\lambda$ appropriately  sweeps out the $R_n(D_k)$ curve. The algorithm stops when $F<\epsilon$. In App. \ref{Bounds} we provide upper and lower bounds, used show that if $F<\epsilon$, we ensure that $|R^k_n(D_k)-R_n(D_k)|<\epsilon$.

Now, let us present a special case and a few extensions for Algorithm \ref{algs}.
\begin{itemize}
\item[(1)] \textit{Regular BAA, i.e., the delay $s=n$}. For delay $s=n$, the algorithm suggested here meets the original BAA, where instead of step (d) we have
    \begin{align}
    r^k(\hat{x}^n|x^n)=\frac{q^{k-1}(\hat{x}^n)2^{-\lambda d(x^n,\hat{x}^n)}}{\sum_{\hat{x}^n}{q^{k-1}(\hat{x}^n)2^{-\lambda d(x^n,\hat{x}^n)}}},\nonumber
    \end{align}
    and in step (e), $q^k(\hat{x}^n)$ corresponds to the joint probability $p(x^n)r^k(\hat{x}^n|x^n)$ as well.
    Moreover, the expression for $c^k_{\hat{x}^n,x^{n-1}}$ is reduced to
    \begin{align}
    c^k_{\hat{x}^n}=\frac{q^{k}(\hat{x}^n)}{q^{k-1}(\hat{x}^n)},\nonumber
    \end{align}
    and the termination of the algorithm in step (g) is defined by
    \begin{align}
    F=\log{\max_{\hat{x}^n}c^k_{\hat{x}^n}}-\sum_{x^n,\hat{x}^n}p(x^n)r^k(\hat{x}^n|x^n)\log{c^k_{\hat{x}^n}}\leq\epsilon,\nonumber
    \end{align}
    as in the regular Blahut-Arimoto algorithm \cite{Bla}.

\item[(2)] \textit{Function of the feed-forward with general delay $s$}. We present a generalization of the algorithm, where the feed-forward is a deterministic function of the source with some delay $s$, $z^{i-s}=f(x^{i-s})$. In that case, step (d) is replaced by
    \begin{align}
    r^k(\hat{x}^n|x^n)=\frac{q^{k-1}(\hat{x}^n||z^{n-s})2^{-\lambda d(x^n,\hat{x}^n)}}{\sum_{\hat{x}^n}q^{k-1}(\hat{x}^n||z^{n-s})2^{-\lambda d(x^n,\hat{x}^n)}},\nonumber
    \end{align}
    and in step (e) we have
    \begin{align}
    q^k(\hat{x}^n||z^{n-s})=\prod_{i=1}^{n}p(\hat{x}_i|\hat{x}^{i-1},z^{i-s}),\nonumber
    \end{align}
    where we calculate $p(\hat{x}_i|\hat{x}^{i-1},z^{i-s})$ from the joint distribution $p(x^n,\hat{x}^n)=p(x^n)r^k(\hat{x}^n|x^n)$.
    The algorithm is terminated in the same way, where
    \begin{align}
    c^k_{\hat{x}^n,z^{n-s}}=\frac{q^{k}(\hat{x}^n||z^{n-s})}{q^{k-1}(\hat{x}^n||z^{n-s})}.\nonumber
    \end{align}
\end{itemize}
\end{subsection}

\begin{subsection}{Complexity and Memory needed}
Computation complexity and memory needed for the algorithm above is presented in Table \ref{ComplexMemory2}.
\begin{table}[h!]
\caption{Memory and operations needed extended BAA for source coding with feed-forward.} 
\centering
\begin{tabular}{|c ||c | c|}
\hline
& Operation & Memory\\
\hline
&&\\
$\min_{p(\hat{x}^n|x^n):\ex{}{d(X^n,\hat{X}^n)}\leq D}\left(\frac{1}{n}I(\hat{X}^n;X^n)\right)$, regular BAA & $O({(|\cal{X}||\hat{\cal{X}}|)}^{n})$ & ${(|\cal{X}||\hat{\cal{X}}|)}^n+{|\cal{X}|}^n+{|\hat{\cal{X}|}}^n$\\
&&\\
\hline
&&\\
$\min_{p(\hat{x}^n|x^n):\ex{}{d(X^n,\hat{X}^n)}\leq D}\left(\frac{1}{n}I(\hat{X}^n\rightarrow X^n)\right)$, Alg. \ref{algs} & $O({(|\cal{X}||\hat{\cal{X}}|)}^{n})$ & $2{(|\cal{X}||\hat{\cal{X}}|)}^n+{|\cal{X}|}^n$\\
&&\\
\hline
\end{tabular}
\label{ComplexMemory2}
\end{table}
\end{subsection}
\end{section}

\begin{section}{DERIVATION OF ALGORITHM \ref{algs}.}\label{SecDer}
In this section, we first describe the alternating minimization procedure, and then (as given in Alg. \ref{algs}) prove its convergence to the global minimum given by
\begin{align}
R_n(D)=\frac{1}{n}\min_{r(\hat{x}^n||x^{n-1}):\ex{}{d(X^n,\hat{X}^n)}\leq D}I(\hat{X}^n\rightarrow X^n).\nonumber
\end{align}
Throughout this section, note that the input probability $p(x^n)$ is fixed in all minimization calculations. Further, we denote by $I_{FF}(r,q)$ the directed information, given by
\begin{align}
I_{FF}(r,q)=\sum_{\hat{x}^n,x^n}{p(x^n)r(\hat{x}^n|x^n)\log{\frac{r(\hat{x}^n|x^n)}{q(\hat{x}^n||x^{n-1})}}}.\nonumber
\end{align}

The alternating maximization procedure is described in \cite{NaissPermuter1} by two maximization functions; $c_2(u_1)\in A_2$  which is the point that achieves $\sup_{u_2\in A_2}f(u_1,u_2)$, and $c_1(u_2)\in A_1$ which is the one that achieves $\sup_{u_1\in A_1}f(u_1,u_2)$. Although in this paper we wish to solve a minimization problem, its negative can be used in the alternating maximization procedure. We now state the alternating maximization procedure lemma.
\begin{lemma}[Lemmas 9.4, 9.5 in \cite{Reymund}, "Convergence of the alternating maximization procedure"]\label{lemconv}.
Let $f(u_1,u_2)$ be a real, concave, bounded from above function, that is continuous and has continuous partial derivatives, and let the sets $A_1,A_2,$ over which we maximize be convex.
Further, assume that $c_2(u_1)\in A_2$ and $c_1(u_2)\in A_1$ for all $u_1\in A_1,\ u_2\in A_2$. Let us define an iteration as the following equation
\begin{align}
(u_1^k,u_2^k)=\left(c_1(u_2^{k-1}),c_2(c_1(u_2^{k-1}))\right),\nonumber
\end{align}
and in each iteration we consider the value $f^k=f(u_1^k,u_2^k)$. Under these conditions, $\lim_{k\rightarrow\infty}f^k=f^{*}$, where $f^*$ is the solution to the optimization problem.
\end{lemma}

The rate-distortion function with feed-forward can be, as in \cite{Bla}, carried out parametrically in terms of parameter $\lambda$, which is introduced as a Lagrange multiplier. In App. \ref{Bounds} we show that this parameter defines the slope of the curve $R_n(D)$ at the point it parameterizes, and the slope is given by $\frac{-\lambda}{n}$. We now write the following parametric expression for $R_n(D)$.
\begin{align}
R_n(D)=\frac{1}{n}\min_{r(\hat{x}^n|x^{n})}\left[I(\hat{X}^n\rightarrow X^n)+\lambda\left(\ex{r}{d(X^n,\hat{X}^n)}-D\right)\right],\label{RD2}
\end{align}
where $D$ is the distortion at the point $r^*(\hat{x}^n|x^n)$ that achieves $R_n(D)$. Here, the value of $D$ is not an input to the minimization, but is determined by the parameter $\lambda$.

Note that the directed information is a function of the joint distribution $p(x^n)r(\hat{x}^n|x^n)$. Since the source distribution is given, the directed information $I_{FF}$ is determined by $r=r(\hat{x}^n|x^n)$ alone. Let us define by $q=q(\hat{x}^n||x^{n-1})$ the causal conditioning probability. Now, let us define the functional
\begin{align}
K(r,q)=I_{FF}(r,q)+\lambda\ex{r}{d(X^n,\hat{X}^n)}.\label{Kfunc1}
\end{align}
From (\ref{RD2}) and (\ref{Kfunc1}) we can see, that $R_n(D)$ can be written as
\begin{align}
R_n(D)=\frac{1}{n}\left[-\lambda D+\min_{r}K(r,q)\right],\nonumber
\end{align}
where $q(\hat{x}^{n-1}||x^n)$ corresponds to the joint distribution $p(x^n)r(\hat{x}^n|x^n)$, and $D$ is the distortion at the point $r^*(\hat{x}^n|x^n)$ that achieves $R_n(D)$.

In this section, we show that we can use the alternating minimization procedure for computing $R_n(D)$. For this purpose, we present several lemmas that assist in proving our main goal. In Lemma \ref{isconvex} we show that the expression we minimize satisfies the conditions in Lemma \ref{lemconv}. In Lemma \ref{whyalts} we show that we are allowed to minimize the functional $K$ over $r(\hat{x}^n|x^n)$ and $q(\hat{x}^n||x^{n-1})$ together, rather than over $r(\hat{x}^n|x^n)$ alone, and thus use the alternating minimization procedure to achieve the optimum value. Lemma \ref{rfixs} is a supplementary claim that helps us to prove Lemma \ref{whyalts}, in which we find an expression for $q(\hat{x}^n||x^{n-1})$ that minimizes the functional $K$ where $r(\hat{x}^n|x^n)$ is fixed. In Lemma \ref{qfixs} we find an explicit expression for $r(\hat{x}^n|x^n)$ that minimizes the functional $K$ where $q(\hat{x}^n||x^{n-1})$ is fixed. Theorem \ref{Ths} combines all lemmas to show that the alternating minimization procedure, as described in Alg. \ref{algs}, converges. We end with a supplementary claim about the upper and lower bounds to the rate distortion, and then prove that the stopping condition described in Alg. \ref{algs} ensures that the error $|R^k_n(D)-R_n(D)|<\epsilon$. From here on, we denote the probabilities over which we minimize as $r=r(\hat{x}^n|x^n),\ q=q(\hat{x}^n||x^{n-1})$.

\begin{lemma}\label{isconvex}
For a fixed input PMF $p(x^n)$, the functional $K$ given in (\ref{Kfunc1}) as a function of $\{r,q\}$ is convex in $\{r,q\}$, continuous and with continuous partial derivatives. Moreover, the sets of probabilities $r,\ q$ (denoted by $A_1,\ A_2$) over which we optimize are convex.

\begin{proof}
Since the functional $K$ consists of a linear (and thus convex) expression in $r$, i.e., $\ex{r}{d(X^n,\hat{X}^n)}$, we only need to verify that the directed information is convex.
We first write the directed information in the following form
\begin{align}
I(\hat{X}^n\rightarrow X^n)&=-\sum_{\hat{x}^n,x^n}{p(x^n,\hat{x}^n)\log{\frac{p(x^n)}{p(x^n||\hat{x}^n)}}}\nonumber\\
&=-\sum_{\hat{x}^n,x^n}{p(x^n,\hat{x}^n)\log{\frac{p(x^n)q(\hat{x}^n||x^{n-1})}{p(x^n||\hat{x}^n)q(\hat{x}^n||x^{n-1})}}}\nonumber\\
&=-\sum_{\hat{x}^n,x^n}{p(x^n,\hat{x}^n)\log{\frac{q(\hat{x}^n||x^{n-1})}{p(x^n,\hat{x}^n)/p(x^n)}}}\nonumber\\
&=-\sum_{\hat{x}^n,x^n}{p(x^n)r(\hat{x}^n|x^n)\log{\frac{q(\hat{x}^n||x^{n-1})}{r(\hat{x}^n|x^n)}}}\nonumber\\
&=I_{FF}(r,q).\nonumber
\end{align}
This form is the negative of a concave function as proven in \cite[Lemma 2]{NaissPermuter1}. Furthermore, in the same lemma we show that the directed information is continuous with continuous partial derivatives; the same explanation applies here. It is also simple to verify that both sets we minimize over are convex, i.e., sets $A_1,\ A_2$, where
\begin{align}
A_1&=\{r(\hat{x}^n|x^n): r(\hat{x}^n|x^n)>0 \text{\ is a regular conditioned PMF}\},\nonumber\\
A_2&=\{q(\hat{x}^n||x^{n-1}): q(\hat{x}^n||x^{n-1}) \text{\ is a causally conditioned PMF}\}.\label{sets2}
\end{align}
\end{proof}
\end{lemma}

Recall that in order to use the alternating minimization procedure we minimize over $\{r(\hat{x}^n|x^n),\ q(\hat{x}^n||x^{n-1})\}$ instead of over $r(\hat{x}^n|x^n)$ alone, and thus need the following lemma.
\begin{lemma}\label{whyalts}
For any discrete random variables $X^n,\ \hat{X}^n$, the following holds
\begin{align}
R_n(D)=\frac{1}{n}\left[-\lambda D+\min_{r,q}K(r,q)\right],\nonumber
\end{align}
where $D$ is the distortion at the point $r^*(\hat{x}^n|x^n)$ that achieves $R_n(D)$

To prove this lemma, we note that $\ex{r}{d(X^n,\hat{X}^n)}$, which does not contain the variable $q$, is part of the functional $K$. Hence, it suffices to show that
\begin{align}
\min_{r(\hat{x}^n|x^n)}\frac{1}{n}I(\hat{X}^n\rightarrow X^n)=
 \min_{q(\hat{x}^n||x^{n-1})}\min_{r(\hat{x}^n|x^n)}\frac{1}{n}I(\hat{X}^n\rightarrow X^n)\label{doublemin}
\end{align}

The proof is given after the following supplementary claim, in which we calculate the specific $q(\hat{x}^n||x^{n-1})$ that minimizes the directed information when $r(\hat{x}^n|x^n)$ is fixed.
\begin{lemma}\label{rfixs}
For fixed $r(\hat{x}^n|x^n)$, there exists a unique $c_2(r)$ that achieves $\min_{q(\hat{x}^n||x^{n-1})}I(\hat{X}^n\rightarrow X^n),$ and is given by
\begin{align}
q^*(\hat{x}^n||x^{n-1})=\frac{p(x^n)r(\hat{x}^n|x^n)}{p(x^n||\hat{x}^n)},\label{QFFform}
\end{align}
where $p(x^n||\hat{x}^n)$ is calculated using the joint distribution $p(x^n)r(\hat{x}^n|x^n)$.
\begin{proof}[Proof for Lemma \ref{rfixs}]
\begin{align}
&I_{FF}(r,q)-I_{FF}(r,q^*)\nonumber\\
&=\sum_{x^n,\hat{x}^n}p(x^n)r(\hat{x}^n|x^n)\log\frac{r(\hat{x}^n|x^n)}{q(\hat{x}^n||x^{n-1})}- \sum_{x^n,\hat{x}^n}p(x^n)r(\hat{x}^n|x^n)\log\frac{r(\hat{x}^n|x^n)}{q^*(\hat{x}^n||x^{n-1})}\nonumber\\
&=\sum_{x^n,\hat{x}^n}p(x^n)r(\hat{x}^n|x^n)\log\frac{q^*(\hat{x}^n||x^{n-1})}{q(\hat{x}^n||x^{n-1})}\nonumber\\
&=\sum_{x^n,\hat{x}^n}p(x^n||\hat{x}^n)q^*(\hat{x}^n||x^{n-1})\log\frac{p(x^n||\hat{x}^n)q^*(\hat{x}^n||x^{n-1})}{p(x^n||\hat{x}^n)q(\hat{x}^n||x^{n-1})}\nonumber\\
&=D\left(p(x^n||\hat{x}^n)q^*(\hat{x}^n||x^{n-1})\parallel p(x^n||\hat{x}^n)q(\hat{x}^n||x^{n-1})\right)\nonumber\\
&\stackrel{(a)}{\geq} 0,\nonumber
\end{align}
where (a) follows from the non-negativity of the divergence. Equality holds if and only if the joint PMFs are the same, i.e., $q=q^*$.
\end{proof}
\end{lemma}
\textit{Proof of Lemma \ref{whyalts}:} The PMF that minimizes the directed information is the one that corresponds to the joint distribution $r(\hat{x}^n|x^n)p(x^n)$; thus (\ref{doublemin}) holds, and thus the functional $K$ can be minimized over both $r,\ q$ combined. \hfill\QED
\end{lemma}

In the following lemma, we derive an explicit expression for $r(\hat{x}^n|x^n)$ that achieves $R_n(D)$, where $q(\hat{x}^n||x^{n-1})$ is fixed.
\begin{lemma}\label{qfixs}
For fixed $q(\hat{x}^n||x^{n-1})$, there exists $c_1(q)$ that achieves $R_n(D)$, and is given by
\begin{equation}
r(\hat{x}^n|x^n)=\frac{q(\hat{x}^n||x^{n-1})2^{-\lambda d(x^n,\hat{x}^n)}}{\sum_{\hat{x}^n}{q(\hat{x}^n||x^{n-1})2^{-\lambda d(x^n,\hat{x}^n)}}}.\nonumber
\end{equation}
\end{lemma}
\begin{proof}
Following \cite[Ch. 5.5.3]{Boyd}, since we are solving a convex optimization problem, we can apply the KKT conditions with the constraints $\sum_{\hat{x}^n}r(\hat{x}^n|x^n)=1$, and set up the functional
\begin{equation}
J=\sum_{x^n,\hat{x}^n}p(x^n)r(\hat{x}^n|x^n)\log\frac{r(\hat{x}^n|x^n)}{q(\hat{x}^n||x^{n-1})}+
\lambda\left(\sum_{x^n,\hat{x}^n}p(x^n)r(\hat{x}^n|x^n)d(x^n,\hat{x}^n)-D\right)+\sum_{x^n}\nu(x^n)\sum_{\hat{x}^n}r(\hat{x}^n|x^n).\nonumber
\end{equation}
Solving $\frac{\partial J}{\partial r(\hat{x}^n|x^n)}=0$ yields the expression for $r(\hat{x}^n|x^n)$ as
\begin{equation}
r(\hat{x}^n|x^n)=\frac{q(\hat{x}^n||x^{n-1})2^{-\lambda d(x^n,\hat{x}^n)}}{\sum_{\hat{x}^n}q(\hat{x}^n||x^{n-1})2^{-\lambda d(x^n,\hat{x}^n)}}.\label{RFFform}
\end{equation}
\end{proof}

Another lemma that is required is one that states that the algorithm, when converges, remains fixed on its variables. we already know that the variable $q$ that optimize the directed information is unique; we have to show that within the algorithm, the variable $r$ is unique as well.
\begin{lemma}\label{runique}
Using the iterations in Alg. \ref{algs}, the variable $r$ is unique, and does not change if convergence is achieved.
\end{lemma}
\begin{proof}
The uniqueness is proven in a similar way to a proof given by Blahut in \cite[Theorem 6]{Bla}, and we follow it with appropriate modifications.
We recall that in the $k$th iteration,
\begin{align}
K(r^k,q^k)&=I_{FF}(r^k,q^k)+\lambda\ex{r^k}{d(X^n,\hat{X}^n)}\nonumber\\
&=\sum_{x^n,\hat{x}^n}p(x^n)r^k(\hat{x}^n|x^n)\log\frac{r^k(\hat{x}^n|x^n)}{q^k(\hat{x}^n||x^{n-1})2^{-\lambda d(x^n,\hat{x}^n)}}.\nonumber
\end{align}
Further, from \cite[Theorem 6]{Bla} we can see that
\begin{align}
K(r^{k+1},q^{k+1})=-\sum_{x^n,\hat{x}^n}p(x^n)r^{k}(\hat{x}^n|x^n)\log\left(\sum_{\hat{x}^n}{q^k(\hat{x}^n||x^{n-1})2^{-\lambda
     d(x^n,\hat{x}^n)}}\right)+\sum_{x^n,\hat{x}^n}p(x^n)r^{k+1}(\hat{x}^n|x^n)\log\frac{q^k(\hat{x}^n||x^{n-1})}{q^{k+1}(\hat{x}^n||x^{n-1})}.\nonumber
\end{align}
Hence,
\begin{align}
K(r^k,q^k)-K(r^{k+1},q^{k+1})&=\sum_{x^n,\hat{x}^n}p(x^n)r^k(\hat{x}^n|x^n)\log\frac{r^k(\hat{x}^n|x^n)\sum_{\hat{x}^n}{q^k(\hat{x}^n||x^{n-1})2^{-\lambda
     d(x^n,\hat{x}^n)}}}{q^k(\hat{x}^n||x^{n-1})2^{-\lambda d(x^n,\hat{x}^n)}}\nonumber\\
     &\ \ \ \ \ +\sum_{x^n,\hat{x}^n}p(x^n)r^{k+1}(\hat{x}^n|x^n)\log\frac{q^{k+1}(\hat{x}^n||x^{n-1})}{q^{k}(\hat{x}^n||x^{n-1})}\nonumber\\
&\stackrel{(a)}{\geq}\sum_{x^n,\hat{x}^n}p(x^n)r^k(\hat{x}^n|x^n)\left(1-\frac{q^k(\hat{x}^n||x^{n-1})2^{-\lambda d(x^n,\hat{x}^n)}}
     {r^k(\hat{x}^n|x^n)\sum_{\hat{x}^n}{q^k(\hat{x}^n||x^{n-1})2^{-\lambda d(x^n,\hat{x}^n)}}}\right)\nonumber\\
     &\ \ \ \ \ +\sum_{x^n,\hat{x}^n}p(x^n)r^{k+1}(\hat{x}^n|x^n)\left(1-\frac{q^k(\hat{x}^n||x^{n-1})}{q^{k+1}(\hat{x}^n||x^{n-1})}\right)\nonumber\\
&\stackrel{(b)}{=}\sum_{x^n,\hat{x}^n}p(x^n)r^k(\hat{x}^n|x^n)\left(1-\frac{r^{k+1}(\hat{x}^n|x^n)}{r^k(\hat{x}^n|x^n)}\right)\nonumber\\
     &\ \ \ \ \ +\sum_{x^n,\hat{x}^n}p(x^n||\hat{x}^n)q^{k+1}(\hat{x}^n||x^{n-1})\left(1-\frac{q^k(\hat{x}^n||x^{n-1})}{q^{k+1}(\hat{x}^n||x^{n-1})}\right)\nonumber\\
&=0+0,\nonumber
\end{align}
where (a) follows from the inequality $\log(y)\geq1-\frac{1}{y}$, and (b) follows from Equation (\ref{RFFform}) where $q=q^k,\ r=r^{k+1}$.
Note, that we have strict inequality unless $q^k=q^{k+1}$, $r^k=r^{k+1}$. Thus, $K(r^k,q^k)$ is non-increasing and is strictly decreasing unless the distribution stabilizes, and hence the uniqueness of the optimum parameter $r^*$ emerges.
\end{proof}

Now, we can prove Theorem \ref{Ths} as stated in Section \ref{SecProbRes}.
\begin{proof}[Proof of Theorem \ref{Ths}]
First, we have to show the existence of a double minimization problem, i.e., an equivalent problem where we minimize over two variables instead of only one; this was shown in Lemma \ref{whyalts}. Now, in order for the alternating minimization procedure to work on this optimization problem, we need to show that the conditions given in Lemma \ref{lemconv} are satisfied for the functional $K$; this was shown in Lemma \ref{isconvex}. The steps described in Alg. \ref{algs} are proved in Lemmas \ref{rfixs} and \ref{qfixs}, thus giving us an algorithm to compute $R_n(D)$, where the minimization is evaluated according to parameter $\lambda$.
\end{proof}

Our last step in proving the convergence of Alg. \ref{algs} is to show why the stopping condition ensures a small error. For this reason
we state a lemma introducing the existence of bounds to the rate distortion with feed-forward function, and then conclude that the stopping condition does ensure a small error in the algorithm, i.e., $|R^k_n(D_k)-R_n(D_k)|<\epsilon$, where $R^k_n(D_k)$ is the upper bound in the $k$th iteration, and $D_k=\ex{r^k}{d(X^n,\hat{X}^n)}$. For this purpose, we define the following expressions in each iteration,
\begin{align}
c^k_{\hat{x}^n,x^{n-1}}&=\frac{q^{k}(\hat{x}^n||x^{n-1})}{q^{k-1}(\hat{x}^n||x^{n-1})}\nonumber\\
\gamma^k(x^n)&=\left(\sum_{\hat{x}^n}q^{k-1}(\hat{x}^n||x^{n-1})2^{-\lambda d(x^n,\hat{x}^n)}\right)^{-1}.\label{cgamk}
\end{align}
\begin{lemma}\label{Lembounds}
Let the parameter $\lambda\geq0$ be given, and let $c^k_{\hat{x}^n,x^{n-1}},\ \gamma^k(x^n)$ be as in (\ref{cgamk}) in the $k$th iteration of Alg. \ref{algs}. Then, at point
\begin{align}
D_k=\ex{r^k}{d(X^n,\hat{X}^n)},\nonumber
\end{align}
we have the following bounds.
\begin{align}
I^k_L(D_k)\leq R_n(D_k)\leq I^k_U(D_k),\nonumber
\end{align}
where
\begin{align}
I^k_U(D_k)&=\frac{1}{n}\left(-\lambda D
    +\sum_{x^n}p(x^n)\log\gamma^k(x^n)-\sum_{x^n,\hat{x}^n}p(x^n)r^k(\hat{x}^n|x^n)\log{c^k_{\hat{x}^n,x^{n-1}}}\right),\nonumber\\
I^k_L(D_k)&=\frac{1}{n}\left(-\lambda D
    +\sum_{x^n}p(x^n)\log\gamma^k(x^n)-\log{\max_{\hat{x}^n,x^{n-1}}c^k_{\hat{x}^n,x^{n-1}}}\right).\label{bounds1}
\end{align}
Note, that $R^k_n(D_k)=I^k_U(D_k)$.
\end{lemma}
The proof for Lemma \ref{Lembounds} is given in App. \ref{Bounds}.

From Lemma \ref{Lembounds} we can conclude the following claim
\begin{corollary}\label{CorStop}
Let us define the error in the algorithm as $|R^k_n(D)-R_n(D)|$. The error defined here is smaller than $\epsilon$ if the following inequality is satisfied:
\begin{align}
F=\log{\max_{\hat{x}^n,x^{n-1}}c^k_{\hat{x}^n,x^{n-1}}}-\sum_{x^n,\hat{x}^n}p(x^n)r^k(\hat{x}^n|x^n)\log{c^k_{\hat{x}^n,x^{n-1}}}\leq\epsilon,\nonumber
\end{align}
where $c^k_{\hat{x}^n,x^{n-1}}$ is defined in the $k$th iteration by Equation (\ref{cgamk}).
\end{corollary}
\begin{proof}
The proof follows from Equation (\ref{bounds1}), in which the upper bound and lower bound differ only in their last expression. Thus, if $F<\epsilon$, then $R_n(D)$ is close to the upper bound $R^k_n(D)$ by, at most, $\epsilon$.
\end{proof}
\end{section}

\begin{section}{Numerical Examples}\label{SecEx}
In this section we present several examples for the rate distortion source coding with feed-forward. First, by using Alg. \ref{algs} we demonstrate, for a specific example, that feed-forward does not decrease the rate distortion function where the source is memoryless (i.i.d.) as shown in\cite{WeisMer}. Then we provide two explicit examples for a Markovian source; one where the distortion is single letter, and one with a general distortion function as presented in \cite{VenPra2}. Geometric programming is used as well, to verify our results.

In all of the examples, we run Alg. \ref{algs} with various values of $\lambda$, and thus construct the graph of $R_n(D)$ using interpolations. Alternatively, one can use the geometric programming form and find, for every distortion $D$ given as input, the rate $R$.
\begin{subsection}{A memoryless (i.i.d.) source}
Analogous to the memoryless channel, it was shown by Weissman and Merhav \cite{WeisMer} that for an i.i.d. source feed-forward does not decrease the rate distortion function. In this example, the source is distributed $X\sim B(\frac{1}{2})$, and the distortion function is single letter, i.e.,
\begin{align}
d(x^n,\hat{x}^n)=\frac{1}{n}\sum_{i=1}^{n}d(x_i,\hat{x}_i).\nonumber
\end{align}
Running our algorithm with delay $s=1$ and block length $n=5$, we would expect to obtain the same result as with no feed-forward at all (as shown in \cite[ch. 10.3.1]{CoverTomas}),which is given by
\begin{equation}
R(D)=\left\{\begin{tabular}[c]{l l} $H_b(p)-H_b(D),$ & $0\leq D\leq\min\{p,1-p\}$\\$0,$ & $D\geq\min\{p,1-p\}$\end{tabular}\right.\label{iidRDformula}
\end{equation}
Note that $H_b(p),\ H_b(D)$ are the binary entropies with parameters $p,\ D$, respectively.
Indeed, the function above and the performance of Alg. \ref{algs} coincide, as illustrated in Fig. \ref{iidRD}.
\begin{figure}[h!]{
\psfrag{R}[][][0.8]{$R(D)$} \psfrag{D}[][][0.8]{$D$}
 \centerline{ \includegraphics[width=6cm]{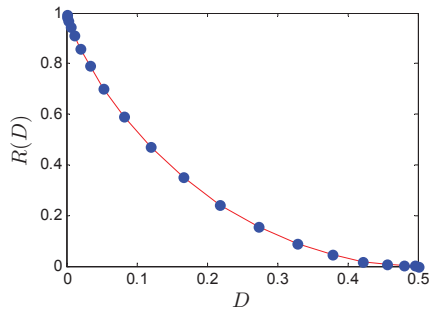}}
  \caption{Rate distortion function for a binary source, and feed-forward with delay 1. The circles represent the performance of Alg. \ref{algs}, regular line is the plot of (\ref{iidRDformula}).}
 \label{iidRD}
}\end{figure}
Note that the joint distribution $p(x^n)r(\hat{x}^n|x^n)$ is the same as the one that achieves the analytical calculation, in which $p(x_i)=0.5$, and $X\oplus\hat{X}\sim B(D)$. For $D=0.2$ and $n=3$, solving the geometrical programming form using a Matlab code produces the rate $R=0.278072$, which is close to $R(0.2)$ using Equation (\ref{iidRDformula}). The value of $\lambda$ turns out to be 6, which means that the slope at point $(R=0.278072, D=0.2)$ is -2.
\end{subsection}

In the following example, we present the performance of Alg. \ref{algs} for a Markov source and a single letter distortion.
\begin{subsection}{Markov source and single letter distortion}
The Markov source is presented in Fig. \ref{StockMod}.
\begin{figure}[h!]{
\psfrag{p}[][][0.8]{$p$} \psfrag{q}[][][0.8]{$q$} \psfrag{P}[][][0.8]{$1-p$} \psfrag{Q}[][][0.8]{$1-q$}\psfrag{1}[][][0.8]{$0$} \psfrag{0}[][][0.8]{$1$}
 \centerline{ \includegraphics[width=6cm]{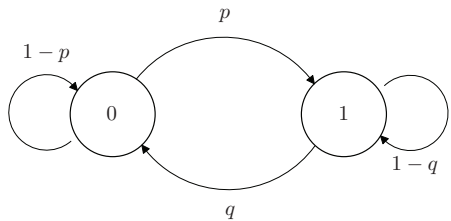}}
  \caption{A symmetrical Markov chain.}
 \label{StockMod}
}\end{figure}
This model was solved by Weissman and Merhav in \cite{WeisMer} for the symmetrical case $p=q$. We extend this model for the case of general transition probabilities $p,q$. The analytical solution for this example is detailed in App. \ref{Appmarkex}; there we show that for any $n$
\begin{align}
R_n(D)=\frac{1}{n}H_b(\pi)+\frac{n-1}{n}\left(\pi_1 H_b(p)+\pi_2 H_b(q)\right)-H_b(D).\label{markp2peq}
\end{align}
By taking $n$ to infinity, we have
\begin{align}
R(D)=\pi_1 H_b(p)+\pi_2 H_b(q)-H_b(D),\nonumber
\end{align}
where $\pi=[\pi_1,\pi_2]$ is the stationary distribution of the source. In Fig. \ref{Markp2p} (a) we present the graphs of $R_n(D)$ for $n=1$ up to $n=12$, where $p=0.3,\ q=0.2$, and $X_0$ has the stationary distribution $[0.4,0.6]$. It is evident that $R_n(D)$ decreases as $n$ increases and converges to the analytical calculation.

In \cite[Lemma 6]{NaissPermuter1} we provided another estimator for the feedback channel capacities, namely, the directed information rate. There, we show that if the limit exists, then
\begin{align}
\lim_{n\rightarrow\infty}\frac{1}{n}I(X^n\rightarrow Y^n)=\lim_{n\rightarrow\infty}\left(I(X^n\rightarrow Y^n)-I(X^{n-1}\rightarrow Y^{n-1})\right).\nonumber
\end{align}
We can also use the directed information rate to estimate $R_n(D)$. This is applied in two ways: either when the rate value is fixed or when the distortion value is fixed. In both cases we first have to fix an axes vector and interpolate the other vector with respect to the fixed one; then we can calculate differences between the interpolated vectors.

In Fig. \ref{Markp2p} (b) we present this estimator only for $n=12$ where the vector of the distortion is interpolated, i.e., $12D_{12}(R)-11D_{11}(R)$. We can see that this estimation is much more accurate than the one in Fig. \ref{Markp2p} (a).
\begin{figure}[h!]{
\begin{center}
\psfrag{R}[][][0.8]{$R(D)$} \psfrag{D}[][][0.8]{$D$} \psfrag{n}[][][0.8]{$n$}
 \subfloat[]{ \includegraphics[width=6cm]{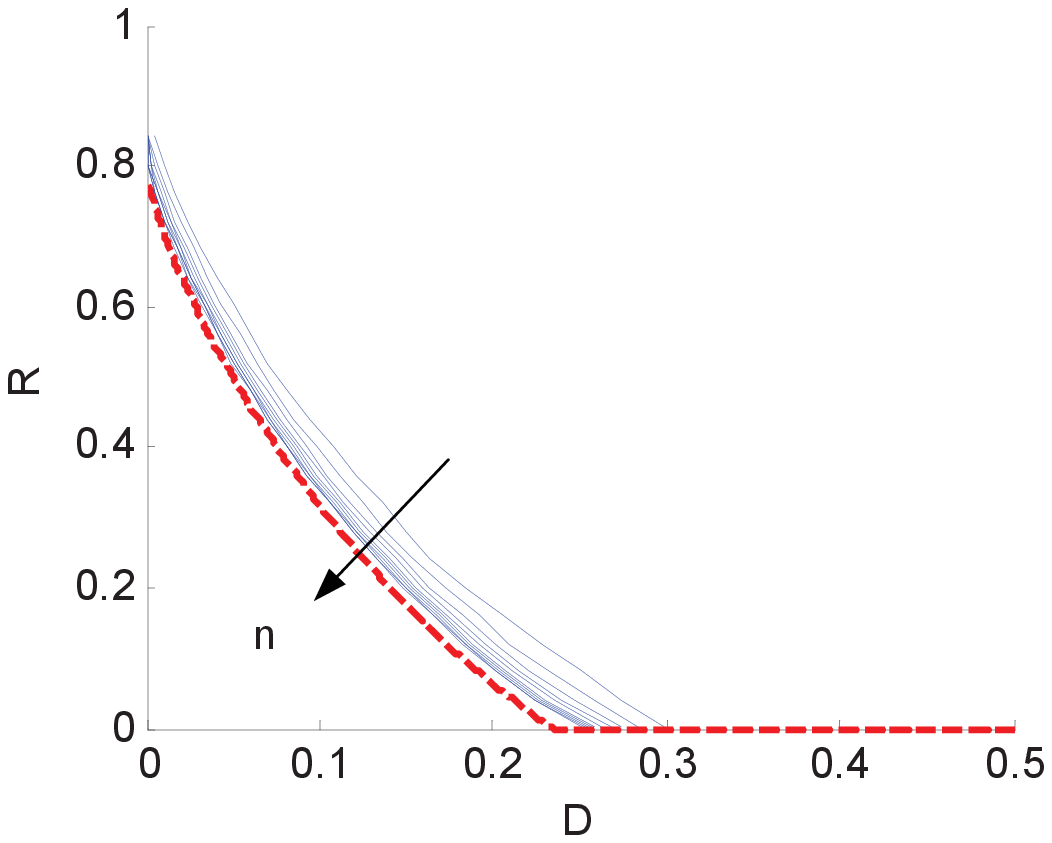}}
 \psfrag{R}[][][0.8]{$R(D)$} \psfrag{D}[][][0.8]{$D$}
 \subfloat[]{ \includegraphics[width=6cm]{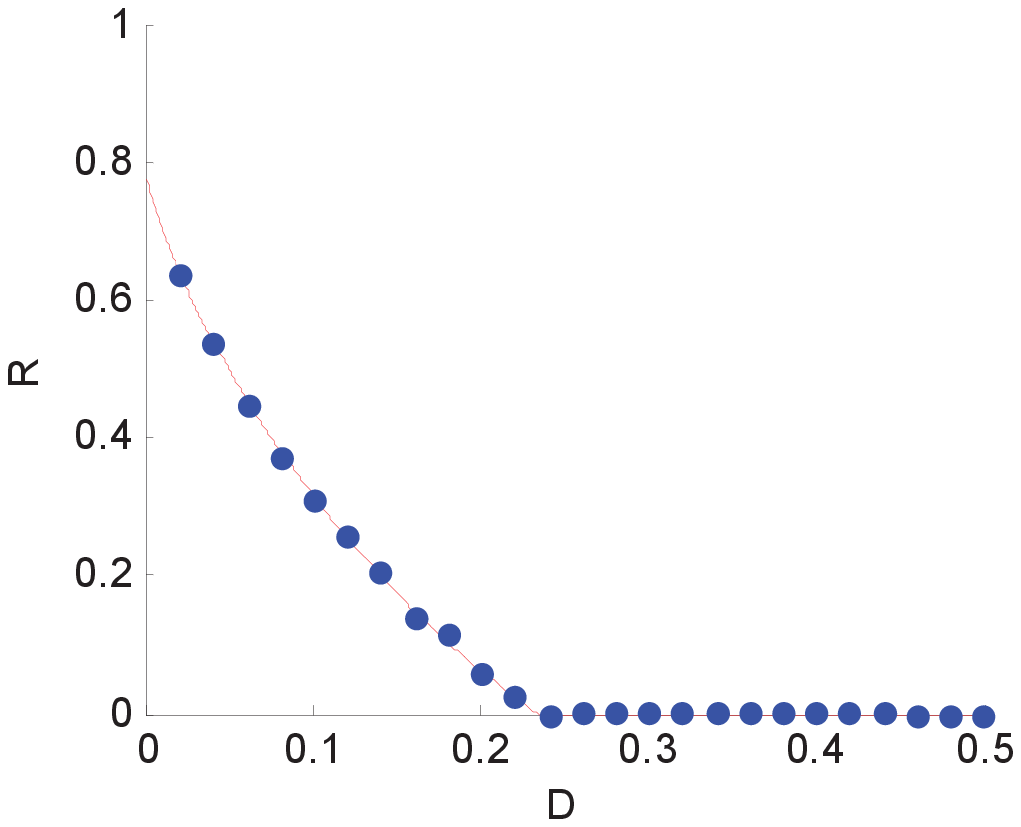}}
  \caption{$R(D)$ for the Markov source example and feed-forward with delay 1.}
    (a) Graph of $R_n(D)$; the arrow marks the way $R_n(D)$ responds to $n$ increasing. The dashed line is the analytical calculation.\\
    (b) Graph of $12D_{12}(R)-11D_{11}(R)$. The circles represent the performance of Alg. \ref{algs}.
  \label{Markp2p}
  \end{center}
}\end{figure}

This is a good opportunity to present the performance of the upper and lower bounds to a specific rate distortion pair $(R,D)$, and the geometrical programming solution to this problem. We ran our BA-type algorithm for the specific parameters $\lambda=9.216,\ n=3$ that corresponds to the rate distortion pair $(R=0.35884, D=0.10627)$ at slope $\frac{9.216}{3}\approx3$, this presented in Fig. \ref{IttGraphmMark} (a). We also ran ten distortion points using GP from $D=0$ to $D=0.27$ and compared it to $R_3(D)$ as in (\ref{markp2peq}) and the BAA performance, the solution is in Fig. \ref{IttGraphmMark} (b).
\begin{figure}[h!]{
\begin{center}
 \psfrag{N}[][][1]{Iteration $k$}\psfrag{V}[][][1]{Value}
 \psfrag{U}[][][1]{$I^k_U$}\psfrag{L}[][][1]{$I^k_L$}
 \subfloat[]{ \includegraphics[width=6cm]{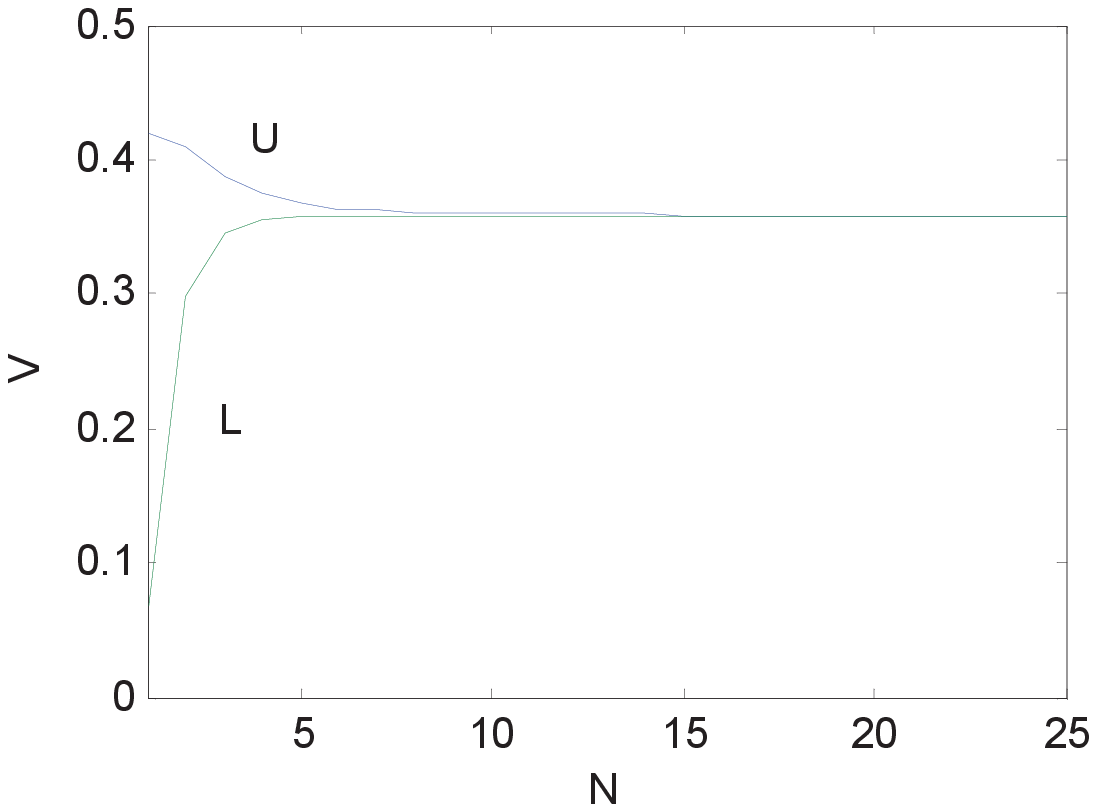}}
 \psfrag{D}[][][1]{$D$} \psfrag{R}[][][1]{$R_3$}
 \subfloat[]{ \includegraphics[width=6cm]{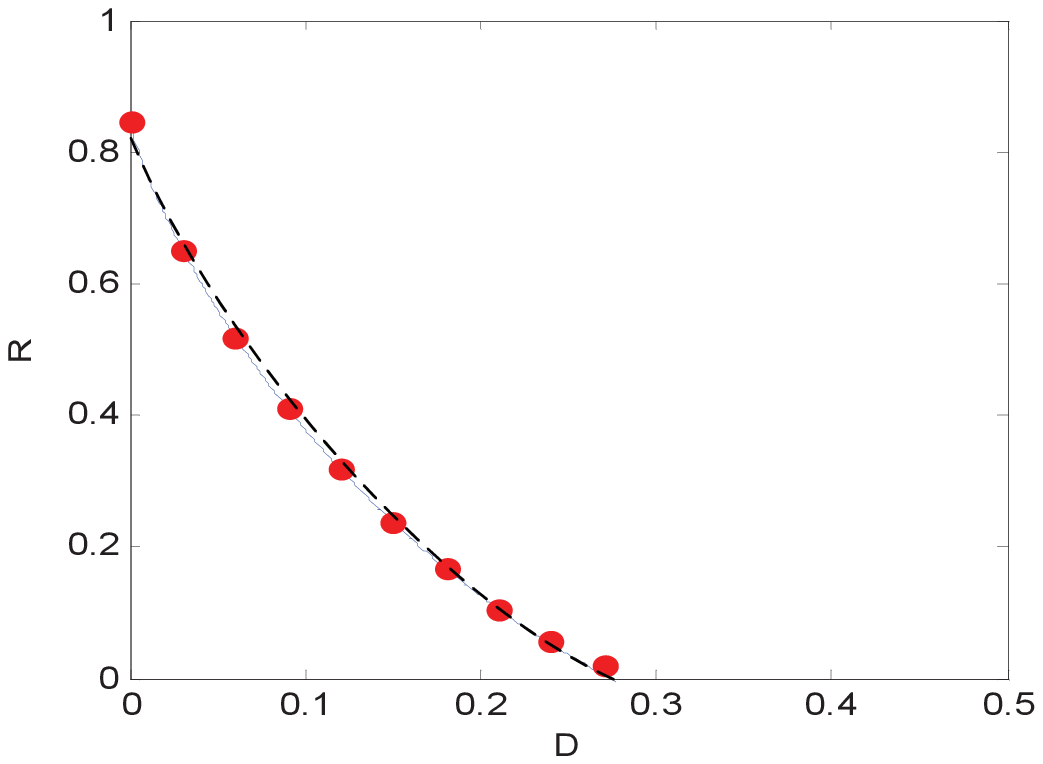}}
 \caption{Bounds for $R_3(D)$ and performance of GP and BAA for $R_3(D)$.}
 (a) Graph of the upper and lower bounds as a function of the iteration for $n=3$, $\lambda=9.216$ as given in Equation(\ref{bounds1}).\\
 (b) Graph of the solution using the GP and BAA method for $n=3$. The solid line is $R_3(D)$ as in (\ref{markp2peq}), the circles represent the performance of the GP, and the dashed line is the BAA result.
 \label{IttGraphmMark}
 \end{center}
}\end{figure}
\end{subsection}

\begin{subsection}{Stock market example. Markov source and general distortion}
The stock market example, in which we wish to observe the behavior of a particular stock over an $N$-day period, was introduced and solved in \cite{VenPra2}. Assume the stock can take $k+1$ values, $0\leq i\leq k$, and is modulated as a $k+1$ state Markov chain. On a given day $i$, the probability for the stock value to increase by 1 is $p_i$, to decrease by 1 is $q_i$, and to remain the same is $1-p_i-q_i$. When the stock value is in state 0, the value cannot decrease.
Similarly, when in state k the value cannot increase. If an investor would like to be forewarned whenever the stock value drops, he is advised with a binary decision $\hat{X}_n$. $\hat{X}_n=1$ if the value drops from day $n-1$ to day $n$, and $\hat{X}_n=0$ otherwise. The distortion is modulated in the following form
\begin{align}
d(x^n,\hat{x}^n)=\frac{1}{n}\sum_{i=1}^{n}e(\hat{x}_i,x_{i-1},x_i),\nonumber
\end{align}
where $e(.,.,.)$ is given in Table \ref{TT3}.
\begin{table}[h!]
\caption{Distortion $e(\hat{x}_i,x_{i-1},x_i)$, $j\in\{0,1,...,k\}$} 
\centering
$(x_{i-1},x_i)$\\
\begin{tabular}{|c |c | c|c|}
\hline
& $j,j+1$ & $j,j$ & $j,j-1$\\
\hline
$\hat{x}_i=0$ & 0 & 0 & 1\\
\hline
$\hat{x}_i=1$ & 1 & 1 & 0\\
\hline
\end{tabular}
\label{TT3}
\end{table}
It was shown in \cite{VenPra2} that the rate-distortion function of a general Markov-chain source with $k$ states, is given by
\begin{align}
R(D)=\sum_{i=1}^{k-1}\pi_i\left(H(p_i,q_i,1-p_i-q_i)-H_b(\epsilon)\right)+\pi_k\left(H_b(q_k)-H_b(\epsilon)\right),\nonumber
\end{align}
where $\pi=[\pi_0,\pi_1,...,\pi_k]$ is the stationary distribution of the Markov chain, and $\epsilon=\frac{D}{1-\pi_0}$.

In our special case we have $k=2$, i.e., $2$ states for the Markov chain, and transition probabilities $p_i=0.3,\ q_i=0.2$ as illustrated in Fig. \ref{StockMod}.
The stationary distribution of such a source is $\pi=[0.4,0.6]$, and we are left with
\begin{align}
R(D)&=\pi_1\left(H_b(q)-H_b(\epsilon)\right)\nonumber\\
&=0.6(H_b(0.2)-H_b(\frac{D}{0.6})).\nonumber
\end{align}
Since the rate cannot be less than zero, and is a descending function of the distortion, the rate-distortion function is as above when $H_b(0.2)\geq H_b(\frac{D}{0.6})$, i.e., when $D\leq 0.12$, and thus we obtain
\begin{equation}
R(D)=\left\{\begin{tabular}[c]{l l} $0.6(H_b(0.2)-H_b(\frac{D}{0.6})),$ & $D\leq0.12$\\$0,$ & otherwise.\end{tabular}\right.\label{analSM}
\end{equation}

In Fig. \ref{SMfig}(a) we present the graphs of $R_n(D)$ for $n=1$ up to $n=12$ with the distortion described here and where $X_0$ has the stationary distribution $[0.4,0.6]$. We can see that $R_n(D)$ decreases as $n$ increases as expected and converges to the analytical calculation.
In Fig. \ref{SMfig} (b) we present the directed information rate estimator only for $n=12$, where the vector of the distortion is interpolated, i.e., $12D_{12}(R)-11D_{11}(R)$. We can see that this estimator is much more accurate than the one in Fig. \ref{SMfig} (a).
\begin{figure}[h!]{
\begin{center}
\psfrag{R}[][][0.8]{$R(D)$} \psfrag{D}[][][0.8]{$D$} \psfrag{n}[][][0.8]{$n$}  \psfrag{A}[][][0.8]{}
 \subfloat[]{ \includegraphics[width=6cm]{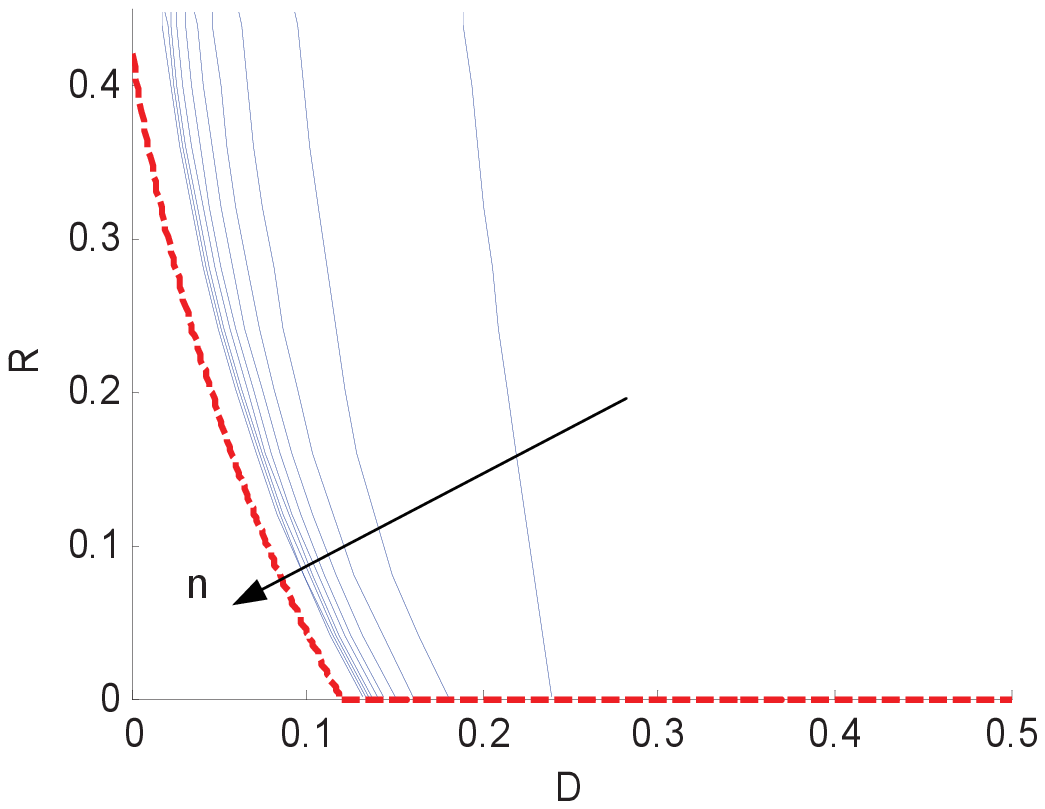}}
 \psfrag{R}[][][0.8]{$R(D)$} \psfrag{D}[][][0.8]{$D$}
 \subfloat[]{ \includegraphics[width=6cm]{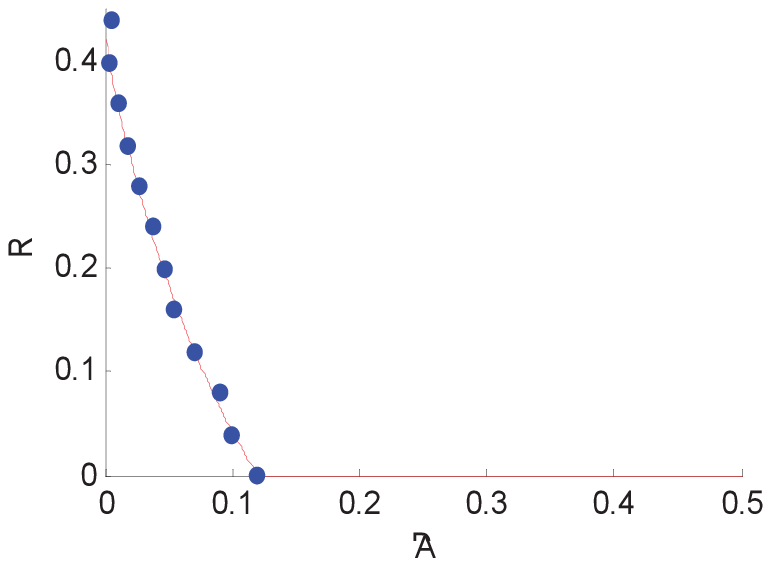}}
  \caption{$R(D)$ for the stock market example and feed-forward with delay 1.} (a) Graph of $R_n(D)$; the arrow marks the way $R_n(D)$ responds to $n$ increasing. The dashed line is the analytical calculation.\\
    (b) Graph of $12D_{12}(R)-11D_{11}(R)$. The circles represent the performance of Alg. \ref{algs}.
  \label{SMfig}
  \end{center}
}\end{figure}
\end{subsection}

\begin{subsection}{The effects of the delay on $R_n(D)$}
In this example we use the Markov source (Fig. \ref{StockMod}) example with a single letter distortion. We run Alg. \ref{algs} with delays $s\in\{1,2,..,10\}$ and block length $n=10$, where $X_0$ has the stationary distribution. We expect the rate distortion function to increase with the delay $s$. This is expected because as the delay $s$ increases the value of the directed information increases as well. Due to the fact that for $s\in\{3,4,...,10\}$ all graphs are close together, we present $R_n(D)$ only for $s=1,2,10$, and the results are shown in Fig. \ref{Markdelay}.
\begin{figure}[h!]{
\psfrag{R}[][][0.8]{$R(D)$} \psfrag{D}[][][0.8]{$D$} \psfrag{d}[][][0.8]{$s\in\{1,2,10\}$}
 \centerline{ \includegraphics[width=6cm]{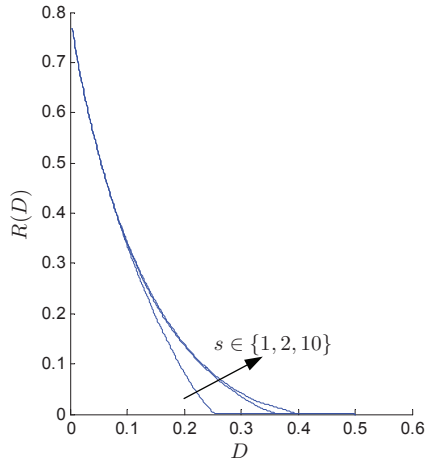}}
  \caption{$R_{10}(D)$ for a Markov source as a function of the delay.}
 \label{Markdelay}
}\end{figure}
\end{subsection}
\end{section}

\begin{section}{Conclusions}
In this paper we considered the rate distortion problem of discrete-time, ergodic, and stationary sources with
feed forward at the receiver. We first derived a sequence of achievable rates, $\{R_n(D)\}_{n\geq1}$, that
converge to the feed-forward rate distortion. By showing that the sequence is sub-additive, we proved that the limit of $R_n(D)$ exists and thus equals to the feed-forward rate distortion.
We provided an algorithm for calculating $R_n(D)$ using the alternating minimization procedure, and also presented a dual form for the optimization of $R_n(D)$, and transformed it into a geometric programming maximization problem.
\end{section}

\appendices
\section{Proof of Lemma \ref{LemSubAdd}}\label{Appsubadd}
We start by showing that the sequence $\{R_n(D)\}$ is sub additive; the methodology is similar to Gallager's proof in \cite[Th. 9.8.1]{Gal} for the case of no feed-forward. Then, by showing that the sequence $R_n(D)$ is sub-additive, following \cite[Lemma 4A.2]{Gal} we obtain our main objective, i.e.,
\begin{align}
\lim_nR_n(D)=\inf_nR_n(D).\nonumber
\end{align}

To commence, we recall that a sequence $\{a_n\}$ is called sub-additive if for all $m,l$,
\begin{align}
(m+l)a_{m+l}\leq ma_m+la_l.\nonumber
\end{align}
Let $l,n$ be arbitrary positive integers and, for a given $D$, let $p_n(\hat{x}^n|x^n)$ and $p_l(\hat{x}^l|x^l)$ be the conditional PMFs that achieve the minimum of the directed information with block length of $n$ and $l$, i.e., that achieve $R_n(D)$ and $R_l(D)$, respectively. Suppose we transmit $m=n+l$ samples as follows; the first $n$ samples are transmitted using $p_n$, and the sequential $l$ samples are transmitted using $p_l$. Hence, the overall conditional PMF is
\begin{align}
p_{n+l}(\hat{x}^{n+l}|x^{n+l})=p_n(\hat{x}^{n}|x^{n})p_{l}(\hat{x}_{n+1}^{n+l}|x_{n+1}^{n+l}).\nonumber
\end{align}
We can see in Section \ref{secalgs} that the directed information can be written as
\begin{align}
I(\hat{X}^m\rightarrow X^m)=H(\hat{X}^m||X^{m-1})-H(\hat{X}^m|X^m).\nonumber
\end{align}
From the construction of the conditional overall PMF $p_{n+l}$, its clear that \begin{align}
H(\hat{X}^{n+l}|X^{n+l})=H(\hat{X}^n|X^n)+H(\hat{X}_{n+1}^{n+l}|X_{n+1}^{n+l}).\nonumber
\end{align}
Furthermore,
\begin{align}
H(\hat{X}^m||X^{m-1})&=\sum_{i=1}^{n+l}H(\hat{X}_i|\hat{X}^{i-1},X^{i-1})\nonumber\\
&=H(\hat{X}^n||X^{n-1})+\sum_{i=n+1}^{n+l}H(\hat{X}_i|\hat{X}^{i-1},X^{i-1})\nonumber\\
&\leq H(\hat{X}^n||X^{n-1})+\sum_{i=n+1}^{n+l}H(\hat{X}_i|\hat{X}_{n+1}^{i-1},X_{n+1}^{i-1})\nonumber\\
&=H(\hat{X}^n||X^{n-1})+H(\hat{X}_{n+1}^{n+l}||X_{n+1}^{n+l-1}).\nonumber
\end{align}
Thus, it follows that
\begin{align}
I(\hat{X}^{n+l}\rightarrow X^{n+l})\leq I(\hat{X}^{n}\rightarrow X^{n})+I(\hat{X}_{n+1}^{n+l}\rightarrow X_{n+1}^{n+l}).\label{upb}
\end{align}
Since the source is stationary, we can start the input block at any given time index; thus the PMFs $p_n$ and $p_l$ achieve $nR_n(D)+lR_l(D)$ on the right-hand side of Equation (\ref{upb}), while the left-hand side is greater than $(n+l)R_{n+l}(D)$ since we attempt to minimize the expression to achieve the rate distortion function. Hence, we obtain
\begin{align}
(n+l)R_{n+l}(D)\leq nR_n(D)+lR_l(D).\nonumber
\end{align}
Using \cite[Lemma 4A.2]{Gal} for sub-additive sequences, we obtain
\begin{align}
\inf_n R_n(D)=\lim_{n\to\infty}R_n(D).\nonumber
\end{align}\hfill\QED

\section{Proof of Lemma \ref{LemConverse}.}\label{Appconv}
In this Appendix we prove Lemma \ref{LemConverse}, which provides for us that the mathematical expression for the rate distortion feed-forward
\begin{align}
R^{(I)}(D)=\lim_{n\to\infty}\frac{1}{n}\min_{p(\hat{x}^n|x^n):\ex{}{d(X^n,\hat{X}^n)}\leq D}I(\hat{X}^n\to X^n),\label{RDI2}
\end{align}
is a lower bound to the operational definition $R(D)$.
\begin{proof}
Consider any $(n,2^{nR},D)$  rate distortion with feed-forward code defined by the mappings $f,\ \{g_i\}_{i=1}^n$ as given in Section \ref{SecProbRes}, Equation (\ref{MapEndDec}), and distortion constraint $\ex{}{d(X^n,\hat{X}^n)}\leq D+\epsilon_n$, where $\epsilon_n\to 0$ as $n$ goes to infinity.
Let the message sent be a random variable $T=f(X^n)$, and assume that the distortion constraint is satisfied. Then we have the following chain of inequalities:
\begin{eqnarray}
nR&\stackrel{(a)}{\geq}&H(T)\nonumber\\
&\geq&I(X^n;T)\nonumber\\
&\stackrel{(b)}{=}&\sum_{i=1}^{n}I(X_i;T|X^{i-1})\nonumber\\
&=&\sum_{i=1}^{n}\left(H(X_i|X^{i-1})-H(X_i|X^{i-1},T)\right)\nonumber\\
&\stackrel{(c)}{=}&\sum_{i=1}^{n}\left(H(X_i|X^{i-1})-H(X_i|X^{i-1},T,\hat{X}^i)\right)\nonumber\\
&\stackrel{(d)}{\geq}&\sum_{i=1}^{n}\left(H(X_i|X^{i-1})-H(X_i|X^{i-1},\hat{X}^i)\right)\nonumber\\
&=&\sum_{i=1}^{n}I(X_i;\hat{X}^i|X^{i-1})\nonumber\\
&\stackrel{(e)}{=}&I(\hat{X}^n\rightarrow X^n),\nonumber
\end{eqnarray}
where (a) follows from the fact that the alphabet of $T$ is $nR$, (b) follows from the chain rule for mutual information, (c) is due to the fact that given $X^{i-1},T$, we know $\hat{X}^i$, and (d) is since conditioning reduces the entropy. Step (e) follows the chain rule for directed information.
Taking $n$ to infinity, we obtain $R\geq R^{(I)}(D)$, and the distortion constraint satisfies
\begin{align}
\lim_{n\to\infty}\ex{}{d(X^n,\hat{X}^n)}\leq D.\nonumber
\end{align}
\end{proof}

\section{Proof of Theorem \ref{ThRdMax}.}\label{Apprdmax}
In this appendix we provide a proof for Theorem \ref{ThRdMax}. We recall that Theorem \ref{ThRdMax} states that the rate distortion function can be written as the following optimization problem:
\begin{align}
R_n(D)=\max_{\lambda\geq0,\gamma(x^n)}\frac{1}{n}
    \left(-\lambda D+\sum_{x^n}p(x^n)\log\gamma(x^n)\right),\label{lowb3}
\end{align}
where, for some causal conditioned probability $p'(x^n||\hat{x}^n)$,  $\gamma(x^n)$ satisfies the inequality constraint
\begin{align}
p(x^n)\gamma(x^n)2^{-\lambda d(x^n,\hat{x}^n)}\leq p'(x^n||\hat{x}^n).\label{cons2}
\end{align}

We prove this theorem in two ways. One is similar to Berger's proof in \cite{Berger}, based on the inequality $\log(y)\geq1-\frac{1}{y}$, for the regular rate distortion function. The other is using the Lagrange duality between the minimization problem we are familiar with and a maximization problem as presented in \cite{Boyd} and \cite{ChiBoyd}. We also provide the connection between the curve of $R_n(D)$ and the parameter $\lambda$; this is embodied in Lemma \ref{Rtag}.

Before we begin, we recall that a step in Alg. \ref{algs} is defined by the following equality
\begin{align}
r^k(\hat{x}^n|x^n)=\frac{q^{k-1}(\hat{x}^n||x^{n-1})2^{-\lambda d(x^n,\hat{x}^n)}}{\sum_{\hat{x}'^n}q^{k-1}(\hat{x}'^n||x^{n-1})2^{-\lambda d(x^n,\hat{x}'^n)}}.\label{step1}
\end{align}
This equality is the outcome of differentiating the Lagrangian when $q(\hat{x}^n||x^{n-1})$ is fixed, as given in Section \ref{SecDer}.
We shall use this equality throughout the proof.

As mentioned, the first proof follows the one in \cite{Berger}.
\begin{proof}[Proof of Theorem \ref{ThRdMax}]
First, we show that for every $r(\hat{x}^n|x^n)$ for which the distortion constraint is satisfied, the following chain of inequalities holds
\begin{align}
I_{FF}(r,q)+\lambda D-\sum_{x^n}p(x^n)\log\gamma(x^n)&\stackrel{(a)}{\geq}
  I_{FF}(r,q)+\lambda\ex{r(\hat{x}^n|x^n)}{d(X^n,\hat{X}^n)}-\sum_{x^n}p(x^n)\log\gamma(x^n)\nonumber\\
&=\sum_{x^n,\hat{x}^n}p(x^n)r(\hat{x}^n|x^n)\log\frac{r(\hat{x}^n|x^n)2^{\lambda
    d(x^n,\hat{x}^n)}}{q(\hat{x}^n||x^{n-1})\gamma(x^n)}\nonumber\\
&\stackrel{(b)}{\geq}\sum_{x^n,\hat{x}^n}p(x^n)r(\hat{x}^n|x^n)
    \left(1-\frac{q(\hat{x}^n||x^{n-1})\gamma(x^n)}{r(\hat{x}^n|x^n)2^{\lambda d(x^n,\hat{x}^n)}}\right)\nonumber\\
&=1-\sum_{x^n,\hat{x}^n}q(\hat{x}^n||x^{n-1})p(x^n)\gamma(x^n)2^{-\lambda d(x^n,\hat{x}^n)}\nonumber\\
&\stackrel{(c)}{\geq}1-\sum_{x^n,\hat{x}^n}q(\hat{x}^n||x^{n-1})p'(x^n||\hat{x}^n)\nonumber\\
&\stackrel{(d)}{=}0,\nonumber
\end{align}
where (a) follows from the fact that the distortion $D$ exceeds $\ex{r(\hat{x}^n|x^n)}{d(X^n,\hat{X}^n)}$ for every $r(\hat{x}^n|x^n)$ as has been assumed, (b) follows from the inequality $\log\frac{1}{y}\geq 1-\frac{1}{y}$, (c) is due to the constraint in Equation (\ref{cons2}), and (d) follows from the fact that $q(\hat{x}^n||x^{n-1})p'(x^n||\hat{x}^n)$ is equal to some joint distribution $p(x^n,\hat{x}^n)$ \cite{Massey}. Since the chain of inequalities is true for every $r(\hat{x}^n|x^n)$, we can choose the one that achieves $R_n(D)$, and then divide by $n$ to obtain the inequality in Equation (\ref{lowb3}) in our Theorem.

To complete the proof of Theorem \ref{ThRdMax}, we need to show that equality holds in the chain of inequalities above for some $\gamma(x^n)$ that satisfies the constraint. If so, let us denote by $r^*(\hat{x}^n|x^n)$ the conditional PMF that achieves $R_n(D)$. Further, we denote by $q^*(\hat{x}^n||x^{n-1})$ the corresponding causal conditioned PMF. Now, consider the following chain of equalities.
\begin{align}
nR_n(D)&=\sum_{x^n,\hat{x}^n}p(x^n)r^*(\hat{x}^n|x^n)\log\frac{r^*(\hat{x}^n|x^n)}{q^*(\hat{x}^n||x^{n-1})}\nonumber\\
&\stackrel{(a)}{=}\sum_{x^n,\hat{x}^n}p(x^n)r^*(\hat{x}^n|x^n)\log\frac{2^{-\lambda d(x^n,\hat{x}^n)}}
    {\sum_{\hat{x}'^n}q^*(\hat{x}'^n||x^{n-1})2^{-\lambda d(x^n,\hat{x}'^n)}}\nonumber\\
&\stackrel{(b)}{=}-\lambda\ex{r^k(\hat{x}^n|x^n)}{d(X^n,\hat{X}^n)}+\sum_{x^n}p(x^n)\log\gamma(x^n)\nonumber\\
&=-\lambda D+\sum_{x^n}p(x^n)\log\gamma(x^n),\nonumber
\end{align}
where (a) is due to a step in the algorithm given by (\ref{step1}), and by the uniqueness of $r^*(\hat{x}^n|x^n)$ in the algorithm, as shown in Lemma \ref{runique}, and (b) follows the expression for $\gamma(x^n)$ given by
\begin{align}
\gamma(x^n)=\left(\sum_{\hat{x}'^n}q^*(\hat{x}'^n||x^{n-1})2^{-\lambda d(x^n,\hat{x}'^n)}\right).
\end{align}

Therefore, we are left with verifying that the $\gamma(x^n)$ above satisfies the constraint:
\begin{align}
p(x^n)\gamma(x^n)2^{-\lambda d(x^n,\hat{x}^n)}&=p(x^n)\frac{2^{-\lambda d(x^n,\hat{x}^n)}}{\sum_{\hat{x}^n}q^*(\hat{x}^n||x^{n-1})2^{-\lambda d(x^n,\hat{x}^n)}}\nonumber\\
&\stackrel{(a)}{=}\frac{p(x^n)r^*(\hat{x}^n|x^n)}{q^*(\hat{x}^n||x^{n-1})}\nonumber\\
&=\frac{p(x^n,\hat{x}^n)}{q^*(\hat{x}^n||x^{n-1})}\nonumber\\
&\stackrel{(b)}{=}p'(x^n||\hat{x}^n),\nonumber
\end{align}
where (a) follows from Equation (\ref{step1}), and (b) is due to the causal conditioning chain rule. Hence, we showed that $R_n(D)$ is the solution to the optimization problem given in Equation (\ref{lowb3}).
\end{proof}

We also present an alternative proof for Theorem \ref{ThRdMax}, this using the Lagrange duality, as in \cite{Boyd}, \cite{ChiBoyd}.
\begin{proof}[Alternative proof for Theorem \ref{ThRdMax}]
Recall that $R_n(D)$ is the result of
\begin{align}
\min_{r(\hat{x}^n|x^n)}\sum_{\hat{x}^n,x^n}{p(x^n)r(\hat{x}^n|x^n)\log{\frac{r(\hat{x}^n|x^n)}{q(\hat{x}^n||x^{n-1})}}},\nonumber
\end{align}
where $q(\hat{x}^n||x^{n-1})$ is defined by $p(x^n)r(\hat{x}^n|x^n)$, subject to the following conditions:
\begin{align}
&\sum_{x^n,\hat{x}^n}p(x^n)r(\hat{x}^n|x^n)d(x^n,\hat{x}^n)\leq D,\nonumber\\
&\forall\ x^n:\ \sum_{\hat{x}^n}r(\hat{x}^n|x^n)=1,\nonumber\\
&\forall\ x^n,\hat{x}^n:\ r(\hat{x}^n|x^n)\geq 0.\nonumber
\end{align}

Let us define the Lagrangian as
\begin{align}
J(r,\lambda,\gamma,\mu)=\sum_{x^n,\hat{x}^n}p(x^n)&r(\hat{x}^n|x^n)\log\frac{r(\hat{x}^n|x^n)}{q(\hat{x}^n||x^{n-1})}+
    \lambda\left(\sum_{x^n,\hat{x}^n}p(x^n)r(\hat{x}^n|x^n)d(x^n,\hat{x}^n)-D\right)\nonumber\\
&+\sum_{x^n}\gamma(x^n)\left(\sum_{\hat{x}^n}r(\hat{x}^n|x^n)-1\right)-\sum_{x^n,\hat{x}^n}\mu(x^n,\hat{x}^n)r(\hat{x}^n|x^n),\nonumber
\end{align}
where $\mu(x^n,\hat{x}^n)\geq0$ for all $x^n,\hat{x}^n$. Differentiating the Lagrangian, $J(r,\lambda,\gamma,\mu)$, over the variable $r(\hat{x}^n|x^n)$, we obtain
\begin{align}
\frac{\partial J}{\partial r(\hat{x}^n|x^n)}&=p(x^n)\log\frac{r(\hat{x}^n|x^n)}{q(\hat{x}^n||x^{n-1})}+\lambda p(x^n)d(x^n,\hat{x}^n)+\gamma(x^n)-\mu(x^n,\hat{x}^n).\nonumber
\end{align}
Solving the equation $\frac{\partial J}{\partial r(\hat{x}^n|x^n)}=0$ in order to find the optimum value, yields the following expression
\begin{align}
r(\hat{x}^n|x^n)=q(\hat{x}^n||x^{n-1})\gamma'(x^n)2^{\frac{\mu(x^n,\hat{x}^n)}{p(x^n)}-\lambda d(x^n,\hat{x}^n)},\label{minimizer}
\end{align}
where $\gamma'(x^n)=2^{-\frac{\gamma(x^n)}{p(x^n)}}$.
Multiplying both sides by $\frac{p(x^n)}{q(\hat{x}^n||x^{n-1})}$ we are left with the constraint
\begin{align}
p(x^n||\hat{x}^n)&=p(x^n)\gamma'(x^n)2^{\frac{\mu(x^n,\hat{x}^n)}{p(x^n)}-\lambda d(x^n,\hat{x}^n)}\nonumber\\
&\geq p(x^n)\gamma'(x^n)2^{-\lambda d(x^n,\hat{x}^n)},\label{cond}
\end{align}
where $p(x^n||\hat{x}^n)$ is induced by $r(\hat{x}^n|x^n)p(x^n)$.

From \cite[Chapter 5.1.3]{Boyd} we know that $g(\lambda,\gamma,\mu)=J(r^*,\lambda,\gamma,\mu)$ is a lower bound to $R_n(D)$. Substituting the minimizer $r(\hat{x}^n|x^n)$ using Equation (\ref{minimizer}), and the condition given by Equation (\ref{cond}) into $J$, we obtain the Lagrange dual function
\begin{equation}
g(\lambda,\gamma')=\left\{\begin{tabular}[c]{l l} $-\lambda D+\sum_{x^n}p(x^n)\log\gamma'(x^n),$ & $p(x^n)\gamma'(x^n)2^{-\lambda d(x^n,\hat{x}^n)}\leq p(x^n||\hat{x}^n)$\\
$-\infty,$ & otherwise.\end{tabular}\right.\label{analSM}
\end{equation}

By making the constraints explicit, and since the minimization problem is convex, we obtain the Lagrange dual problem, i.e., $R_n(D)$ is the solution to
\begin{align}
\max_{\gamma(x^n),\lambda}\frac{1}{n}\left(-\lambda D+\sum_{x^n}p(x^n)\log\gamma(x^n)\right),\label{maxGP}
\end{align}
subject to
\begin{align}
\forall\ x^n,\hat{x}^n:\ &p(x^n)\gamma(x^n)2^{-\lambda d(x^n,\hat{x}^n)}\leq p(x^n||\hat{x}^n),\nonumber\\
&\lambda\geq0\nonumber
\end{align}
for the $p(x^n||\hat{x}^n)$ that is induced by $r(\hat{x}^n|x^n)p(x^n)$, and $r(\hat{x}^n|x^n)$ is the optimal PMF.

We use the notation of an \textit{optimal} PMF if it achieves the optimal value. For example, the PMF $r(\hat{x}^n|x^n)$ that achieves the minimum of the directed information given the distortion constraint, is optimal. we say that the PMF, $p(x^n||\hat{x}^n)$ is optimal, if it is induced by the optimal $r(\hat{x}^n|x^n)$. Another example is the maximization problem in (\ref{maxGP}). We say that $\lambda, \gamma(x^n)$ are optimal if they achieve the maximum value. Therefore, $p(x^n||\hat{x}^n)$ is optimal as well if it satisfies Equation (\ref{cond}).

Now, we wish to substitute the constraint to
\begin{align}
\forall\ x^n,\hat{x}^n:\ &p(x^n)\gamma(x^n)2^{-\lambda d(x^n,\hat{x}^n)}\leq p'(x^n||\hat{x}^n),\label{condsome}
\end{align}
for some $p'(x^n||\hat{x}^n)$. First, note that we always achieve equality in (\ref{condsome}) since we can increase the value of $\gamma(x^n)$ and thus increase the objective. This, combined with the fact that for $r(\hat{x}^n|x^n)>0$, $\mu(x^n,\hat{x}^n)$ must be zero, we have equality in (\ref{cond}) as well (if $r(\hat{x}^n|x^n)=0$, then $q(\hat{x}^n||x^{n-1})=0$, and Equation (\ref{minimizer}) holds too). Now, let us assume that the maximum in (\ref{maxGP}) with the constraint in (\ref{condsome}) is achieved at a \textit{non-optimal} $p'(x^n||\hat{x}^n)$, i.e., one that is not achieved using the optimal $\lambda, \gamma(x^n)$. Thus, the value obtained in (\ref{maxGP}) is larger then the value achieved by $p(x^n||\hat{x}^n)$, i.e., $R_n(D)$ (since the maximization includes $p(x^n||\hat{x}^n)$). However, from the lagrange duality it should be a lower bound to $R_n(D)$, thus contradicting the fact that the maximum is achieved at a non-optimal $p'(x^n||\hat{x}^n)$.
\end{proof}

Note, that we can construct the optimal PMF $r(\hat{x}^n|x^n)$ from the solution to the maximization problem presented here. Consider the parameters $\lambda,\ \gamma(x^n),$ that achieve (\ref{maxGP}), and calculate $p(x^n||\hat{x}^n)$ according to Equation (\ref{cond}). The calculation of $r(\hat{x}^n|x^n)$ is done recursively on $r(\hat{x}^i|x^i)$. For $i=1$, calculate $r(\hat{x}^1|x^1)$ using
\begin{align}
r(\hat{x}^1|x^1)=\frac{p(x^1||\hat{x}^1)}{p(x^1)}\sum_{x_1}p(x^1)r(\hat{x}^1|x^1).\nonumber
\end{align}
Further, calculate $q(\hat{x}_1)$ using
\begin{align}
q(\hat{x}_1)=\sum_{x_1}p(x^1)r(\hat{x}^1|x^1).\nonumber
\end{align}
Now, once we have $r(\hat{x}^j|x^j),\ q(\hat{x}_j|\hat{x}^{j-1}x^{j-1})$ for every $j<i$, calculate $r(\hat{x}^i|x^i)$ using
\begin{align}
r(\hat{x}^i|x^i)&=\frac{p(x^i||\hat{x}^i)}{p(x^i)}\left[\prod_{j=1}^{i-1}q(\hat{x}_j|\hat{x}^{j-1}x^{j-1})\right]
    \frac{\sum_{x_i}p(x^i)r(\hat{x}^i|x^i)}{p(x^{i-1})r(\hat{x}^{i-1}|x^{i-1})},\nonumber
\end{align}
and then
\begin{align}
q(\hat{x}_i|\hat{x}^{i-1}x^{i-1})&=\frac{\sum_{x_i}p(x^i)r(\hat{x}^i|x^i)}{p(x^{i-1})r(\hat{x}^{i-1}|x^{i-1})}.\nonumber
\end{align}
Do so until $i=n$, and we obtain our optimal $r(\hat{x}^n|x^n)$.

Another lemma we wish to provide is the connection between the curve of $R_n(D)$ and the parameter $\lambda$. This lemma is similar to the one given by Berger in \cite[Th. 2.5.1]{Berger} for the case of no feed-forward.
\begin{lemma}\label{Rtag}
Consider the expression for $R_n(D)$ given by
\begin{align}
R_n(D)=\frac{1}{n}
    \left(-\lambda D+\sum_{x^n}p(x^n)\log\gamma(x^n)\right),\nonumber
\end{align}
where $\gamma(x^n)$ and $\lambda$ are the variables that maximize (\ref{maxGP}). We have seen that $\gamma(x^n)$ is of the form
\begin{align}
\gamma(x^n)=\left(\sum_{\hat{x}^n}q^*(\hat{x}^n||x^{n-1})2^{-\lambda d(x^n,\hat{x}^n)}\right)^{-1}.\nonumber
\end{align}
Hence, the slope at distortion $D$ is $R_n'(D)=-\frac{\lambda}{n}$.
\end{lemma}
\begin{proof}
The proof is given simply by differentiating the expression for $R_n(D)$.
\begin{align}
\frac{dR_n}{dD}&=\frac{\partial R_n}{\partial D}+\frac{\partial R_n}{\partial\lambda}\frac{d\lambda}{dD}+
    \sum_{x^n}\frac{\partial R_n}{\partial\gamma(x^n)}\frac{d\gamma(x^n)}{dD}\nonumber\\
&=\frac{1}{n}\left[-\lambda-D\frac{d\lambda}{dD}+\sum_{x^n}\frac{p(x^n)}{\gamma(x^n)}\frac{d\gamma(x^n)}{dD}\right]\nonumber\\
&=-\frac{\lambda}{n}+\frac{1}{n}\left[-D+\sum_{x^n}\frac{p(x^n)}{\gamma(x^n)}\frac{d\gamma(x^n)}{d\lambda}\right]\frac{d\lambda}{dD}.\nonumber
\end{align}
Now, consider the following expression
\begin{align}
F=\sum_{x^n,\hat{x}^n}p(x^n)q^*(\hat{x}^n||x^{n-1})\gamma(x^n)2^{-\lambda d(x^n,\hat{x}^n)}.\nonumber
\end{align}
Using the $\gamma(x^n)$ given above, we have $F=1$ and thus $\frac{\partial F}{\partial\lambda}=0$. However,
\begin{align}
\frac{\partial F}{\partial\lambda}&=\sum_{x^n,\hat{x}^n}\left[\frac{d\gamma(x^n)}{d\lambda}-d(x^n,\hat{x}^n)\gamma(x^n)\right]
    p(x^n)q^*(\hat{x}^n||x^{n-1})2^{-\lambda d(x^n,\hat{x}^n)}\nonumber\\
&=\sum_{x^n}\frac{d\gamma(x^n)}{d\lambda}p(x^n)\sum_{\hat{x}^n}q^*(\hat{x}^n||x^{n-1})2^{-\lambda d(x^n,\hat{x}^n)}
    -\sum_{x^n,\hat{x}^n}p(x^n)q^*(\hat{x}^n||x^{n-1})2^{-\lambda d(x^n,\hat{x}^n)}\gamma(x^n)d(x^n,\hat{x}^n)\nonumber\\
&=\sum_{x^n}\frac{d\gamma(x^n)}{d\lambda}\frac{p(x^n)}{\gamma(x^n)}-\sum_{x^n,\hat{x}^n}p(x^n)r^*(\hat{x}^n|x^n)d(x^n,\hat{x}^n)\nonumber\\
&=\sum_{x^n}\frac{d\gamma(x^n)}{d\lambda}\frac{p(x^n)}{\gamma(x^n)}-D\nonumber\\
&=0.\nonumber
\end{align}
Hence, we can conclude that
\begin{align}
\frac{dR_n}{dD}&=-\frac{\lambda}{n}+\frac{1}{n}\left[-D+\sum_{x^n}\frac{p(x^n)}{\gamma(x^n)}\frac{d\gamma(x^n)}{d\lambda}\right]\frac{d\lambda}{dD}\nonumber\\
&=-\frac{\lambda}{n}.\nonumber
\end{align}
\end{proof}

\section{Proof for Lemma \ref{Lembounds}}\label{Bounds}
In this appendix we prove the existence of a sequence of upper and lower bounds to $R_n(D)$, the rate distortion function with feed-forward. These bounds correspond to an iteration in Alg. \ref{algs}, and both converge to $R_n(D)$. To this end, we present and prove a few supplementary claims that assist in obtaining our main goal. Theorem \ref{ThRdMax} provides an alternating form (Lagrange dual form) of an optimization problem achieving $R_n(D)$, that is proved in App \ref{Apprdmax}. In Lemma \ref{lemLow2}, we show that in each iteration we can obtain measures that satisfy the constraint in Theorem \ref{ThRdMax} to form a lower bound, and that the bound is tight and achieved as the upper bound converges. We also provide a proof for the existence of a an upper bound in each iteration.

Before we begin, we recall that a step in Alg. \ref{algs} is defined by the following equality
\begin{align}
r^k(\hat{x}^n|x^n)=\frac{q^{k-1}(\hat{x}^n||x^{n-1})2^{-\lambda d(x^n,\hat{x}^n)}}{\sum_{\hat{x}'^n}q^{k-1}(\hat{x}'^n||x^{n-1})2^{-\lambda d(x^n,\hat{x}'^n)}}.\label{step}
\end{align}
We shall use this equality throughout the proof.

As mentioned, we use Theorem \ref{ThRdMax} that provides us with the following alternating optimization problem.
\begin{align}
R_n(D)=\max_{\lambda\geq0,\gamma(x^n)}\frac{1}{n}
    \left(-\lambda D+\sum_{x^n}p(x^n)\log\gamma(x^n)\right),\label{lowb1}
\end{align}
where $\gamma(x^n)$ satisfies the inequality constraint
\begin{align}
p(x^n)\gamma(x^n)2^{-\lambda d(x^n,\hat{x}^n)}\leq p'(x^n||\hat{x}^n)\label{cons1}
\end{align}
for some causal conditioned probability $p'(x^n||\hat{x}^n)$.

We now show that in each iteration in Alg. \ref{algs}, choosing $\gamma(x^n)$ appropriately forms a lower bound for $R_n(D)$.
\begin{lemma}\label{lemLow2}
In the $k$th iteration in Alg. \ref{algs}, by letting
\begin{align}
\gamma'^k(x^n)=\left(\sum_{\hat{x}^n}q^{k-1}(\hat{x}^n||x^{n-1})2^{-\lambda d(x^n,\hat{x}^n)}\right)^{-1},\label{gammak1}
\end{align}
and
\begin{align}
c^k_{\hat{x}^n,x^{n-1}}=\frac{q^{k}(\hat{x}^n||x^{n-1})}{q^{k-1}(\hat{x}^n||x^{n-1})},\label{ck1}
\end{align}
and defining
\begin{align}
\gamma^k(x^n)=\frac{\gamma'^k(x^n)}{\max_{\hat{x}^n,x^{n-1}}c^k_{\hat{x}^n,x^{n-1}}},\label{gammak3}
\end{align}
the constraint in Equation (\ref{cons1}) is satisfied, and forms a lower bound given by
\begin{align}
R_n(D)&\geq\frac{1}{n}\left(-\lambda D
    +\sum_{x^n}p(x^n)\log\gamma^k(x^n)-\log{\max_{\hat{x}^n,x^{n-1}}c^k_{\hat{x}^n,x^{n-1}}}\right).\nonumber
\end{align}
Furthermore, this lower bound is tight, and is achieved as $R^k_n(D)$ converges to $R_n(D)$, where $R^k_n(D)$ is the upper bound.
\end{lemma}
\begin{proof}
Let us fix the parameter $\gamma'^k(x^n)$ as in (\ref{gammak1}). Hence,
\begin{align}
p(x^n)\gamma'^k(x^n)2^{-\lambda d(x^n,\hat{x}^n)}&=p(x^n)\frac{2^{-\lambda d(x^n,\hat{x}^n)}}{\sum_{\hat{x}^n}q^{k-1}(\hat{x}^n||x^{n-1})2^{-\lambda d(x^n,\hat{x}^n)}}\nonumber\\
&\stackrel{(a)}{=}\frac{p(x^n)r^k(\hat{x}^n|x^n)}{q^{k-1}(\hat{x}^n||x^{n-1})}\nonumber\\
&\stackrel{(b)}{=}\frac{p'(x^n||\hat{x}^n)q^k(\hat{x}^n||x^{n-1})}{q^{k-1}(\hat{x}^n||x^{n-1})}\nonumber\\
&\leq p'(x^n||\hat{x}^n)\max_{\hat{x}^n,x^{n-1}}\frac{q^{k}(\hat{x}^n||x^{n-1})}{q^{k-1}(\hat{x}^n||x^{n-1})}\nonumber
\end{align}
where (a) follows from the definition of a step in Alg. \ref{algs} and given above in Equation (\ref{step}), and (b) follow the chain rule of causal conditioning, and $p'(x^n||\hat{x}^n)=\frac{p(x^n)r^k(\hat{x}^n|x^n)}{q^{k}(\hat{x}^n||x^{n-1})}$ is a causal conditioned PMF.
Hence, combined with (\ref{gammak3}), we obtain
\begin{align}
p(x^n)\gamma^k(x^n)2^{-\lambda d(x^n,\hat{x}^n)}&=\frac{p(x^n)\gamma'(x^n)2^{-\lambda d(x^n,\hat{x}^n)}}
    {\max_{\hat{x}^n,x^{n-1}}c^k_{\hat{x}^n,x^{n-1}}}\nonumber\\
&\leq p'(x^n||\hat{x}^n).\nonumber
\end{align}
Thus, we can use Theorem \ref{ThRdMax}, and obtain a lower bound for $R_n(D)$, i.e.,
\begin{align}
R_n(D)&\geq\frac{1}{n}\left[-\lambda D+\sum_{x^n}p(x^n)\log\gamma^k(x^n)\right]\nonumber\\
&=\frac{1}{n}\left[-\lambda D+\sum_{x^n}p(x^n)\log\gamma'^k_{x^n}
    -\sum_{x^n}p(x^n)\log\left(\max_{\hat{x}^n,x^{n-1}}c^k_{\hat{x}^n,x^{n-1}}\right)\right]\nonumber\\
&=\frac{1}{n}
    \left[-\lambda D+\sum_{x^n}p(x^n)\log\gamma'^k(x^n)-\log\left(\max_{\hat{x}^n,x^{n-1}}c^k_{\hat{x}^n,x^{n-1}}\right)\right].\label{lowb2}
\end{align}

To complete the proof of this lemma, we are left to show that as $k$ increases, i.e., the upper bound converges to $R_n(D)$, the lower bound is tight. For that matter, we note that the PMFs that achieve the optimum value $q^*,\ r^*$ are unique, as shown in Lemma \ref{runique}. Thus, it is clear that
\begin{align}
c^*_{\hat{x}^n,x^{n-1}}&=\frac{q^*(\hat{x}^n||x^{n-1})}{q^*(\hat{x}^n||x^{n-1})}=1,\label{ck2}
\end{align}
and
\begin{align}
\gamma^k(x^n)=\gamma'^k(x^n)=\left(\sum_{\hat{x}^n}q^{*}(\hat{x}^n||x^{n-1})2^{-\lambda d(x^n,\hat{x}^n)}\right)^{-1}.\label{gammak2}
\end{align}
Placing Equation (\ref{gammak2}) and (\ref{ck2}) in Equation (\ref{lowb2}), as shown in Theorem \ref{ThRdMax}, achieves equality instead of the chain of inequalities given. Thus $R_n(D)$ is, in fact, the solution to the optimization problem given in Equation (\ref{lowb1}), and we have demonstrated the existence of the lower bound
\end{proof}

\begin{lemma}\label{lemLow3}
In the $k$th iteration in Alg. \ref{algs}, the upper bound to the rate distortion is given by
\begin{align}
R_n(D_k)&\leq\frac{1}{n}\left(-\lambda D_k
    +\sum_{x^n}p(x^n)\log\gamma^k(x^n)-\sum_{x^n}p(x^n)r^k(\hat{x}^n|x^n)\log{c^k_{\hat{x}^n,x^{n-1}}}\right),\nonumber
\end{align}
where $D_k=\ex{r^k}{d(X^n,\hat{X}^n)}$.
\end{lemma}
\begin{proof}
Note, that if $r^k(\hat{x}^n,x^n)$ produces a distortion $D$, then
\begin{align}
nR_n(D)&\leq I_{FF}(r^k,q^k)\nonumber\\
&=\sum_{x^n,\hat{x}^n}p(x^n)r^k(\hat{x}^n|x^n)\log\frac{r^k(\hat{x}^n|x^n)}{q^k(\hat{x}^n||x^{n-1})}\nonumber\\
&\stackrel{(a)}{=}\sum_{x^n,\hat{x}^n}p(x^n)r^k(\hat{x}^n|x^n)\log\frac{q^{k-1}(\hat{x}^n||x^{n-1})2^{-\lambda d(x^n,\hat{x}^n)}}
    {q^{k}(\hat{x}^n||x^{n-1})\sum_{\hat{x}'^n}q^{k-1}(\hat{x}'^n||x^{n-1})2^{-\lambda d(x^n,\hat{x}'^n)}}\nonumber\\
&=-\lambda\ex{r^k}{d(X^n,\hat{X}^n)}-\sum_{x^n}p(x^n)\log\sum_{\hat{x}'^n}q^{k-1}(\hat{x}'^n||x^{n-1})2^{-\lambda d(x^n,\hat{x}'^n)}-
    \sum_{x^n,\hat{x}^n}p(x^n)r^k(\hat{x}^n|x^n)\log\frac{q^{k}(\hat{x}^n||x^{n-1})}{q^{k-1}(\hat{x}^n||x^{n-1})}\nonumber\\
&\stackrel{(b)}{=}-\lambda D_k+\sum_{x^n}p(x^n)\log\gamma^k(x^n)-
    \sum_{x^n,\hat{x}^n}p(x^n)r^k(\hat{x}^n|x^n)\log{c^k_{\hat{x}^n,x^{n-1}}},\label{lastineq}
\end{align}
where (a) follows from the definition of a step in Alg. \ref{algs} and is given above in Equation (\ref{step}), and (b) follows from the definition of $\gamma^k(x^n),\ c^k_{\hat{x}^n,x^{n-1}}$. Hence, we have formed an upper bound to the rate distortion as in the lemma. Note that the only inequality is in the first line of the chain, and is due to the fact that $I_{FF}(r^k,q^k)\geq\min_{r,q}I_{FF}(r,q)$. However, upon convergence, this inequality is tight.
\end{proof}

We can now conclude our main objective in this appendix.

\textit{Proof of Lemma \ref{Lembounds}} Proving this lemma requires us to present upper and lower bounds that converge to $R_n(D)$. Lemma \ref{lemLow2} provides us with a lower bound and its tightness, whereas Lemma \ref{lemLow3} provides us with a tight upper bound as well, as required.
\hfill\QED

\section{Solution to $R(D)$ for an asymmetrical Markov source.}\label{Appmarkex}
The Markov source is presented in Fig. \ref{StockMod} above. We can describe the process $\{X_i\}$ using the equation
\begin{align}
X_i&=X_{i-1}W_1+(1-X_{i-1})W_2\nonumber\\
&=(X_{i-1}(W_1\oplus W_2))\oplus W_2,\nonumber
\end{align}
where $W_1\sim B(q)$, $W_2\sim B(p)$. This allows us to evaluate $H(X_n|X_{n-1})$:
\begin{align}
H(X_n|X_{n-1})&=H((X_{n-1}(W_1\oplus W_2))\oplus W_2|X_{n-1})\nonumber\\
&=p(x_{n-1}=1)H(W_1\oplus W_2\oplus W_2)+p(x_{n-1}=0)H(W_2)\nonumber\\
&=\pi_1 H(W_1)+\pi_2 H(W_2),\nonumber
\end{align}
where $\pi$ is the stationary distribution of the source. Now, to find the rate distortion of this model, we start with the converse
\begin{align}
\frac{1}{n}I(\hat{X}^n\rightarrow X^n)&=H(X^n)-H(X^n||\hat{X}^n)\nonumber\\
&=\frac{1}{n}H(X_1)+\frac{n-1}{n}H(X_n|X_{n-1})-\frac{1}{n}\sum_{i=1}^nH(X_i|X^{i-1},\hat{X}^i)\nonumber\\
&\stackrel{(a)}{\geq} \frac{1}{n}H_b(\pi)+\frac{n-1}{n}H(X_n|X_{n-1})-\frac{1}{n}\sum_{i=1}^nH(X_i|\hat{X}_i)\nonumber\\
&\stackrel{(b)}{\geq} \frac{1}{n}H_b(\pi)+\frac{n-1}{n}H(X_n|X_{n-1})-H_b(D)\nonumber\\
&=\frac{1}{n}H_b(\pi)+\frac{n-1}{n}\left(\pi_1 H_b(p)+\pi_2 H_b(q)\right)-H_(D),\nonumber
\end{align}
where (a) follows from the fact that conditioning reduces entropy, and (b) follows the fact that $P(X_i\neq\hat{X}_i)\leq D$ and $H_b(D)$ increases with $D$ for $D\leq\frac{1}{2}$.

However, we can achieve it by letting $X_i$ depend on $\hat{X}_i$ and $X_{i-1}$ as in Fig. \ref{MarkEx1},
\begin{figure}[h!]{
\psfrag{p}[][][1]{$1-p_1$} \psfrag{D}[][][1]{$D$}\psfrag{q}[][][1]{$p_1$} \psfrag{E}[][][1]{$1-D$}\psfrag{p}[][][1]{$1-p_1$} \psfrag{r}[][][1]{$p_2$}\psfrag{s}[][][1]{$1-p_2$}
\psfrag{X}[][][1]{$X_{i-1}$}\psfrag{Y}[][][1]{$\hat{X}_i$} \psfrag{Z}[][][1]{$X_i$}
 \centerline{ \includegraphics[width=6cm]{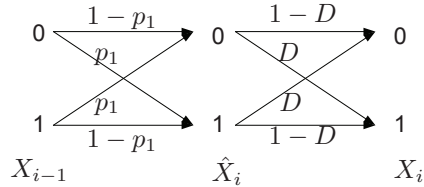}}
 \caption{Distribution of $X_i$ given $X_{i-1}$ and $\hat{X}_i$.}
 \label{MarkEx1}
}\end{figure}
where $p_1,\ p_2$ must hold for the following equation
\begin{align}
p_1D+(1-p_1)(1-D)=1-p,\nonumber\\
p_2D+(1-p_2)(1-D)=1-q,\nonumber
\end{align}
i.e.,
\begin{align}
p_1=\frac{D-p}{2D-1},\nonumber\\
p_2=\frac{D-q}{2D-1}.\nonumber
\end{align}
Note, that under this construction, the source $X^n$ is still Markovian. Further, from Fig. \ref{MarkEx1} we can see that  $X_{i-1}-\hat{X}_i-X_i$ forms a Markov chain, and $H(X_i|\hat{X}_i)=H_b(D)$. Thus, we obtain equality in (a), (b) in the above chain of inequalities, and hence showed that
\begin{align}
R_n(D)=\frac{1}{n}H_b(\pi)+\frac{n-1}{n}\left(\pi_1 H_b(p)+\pi_2 H_b(q)\right)-H_b(D).\nonumber
\end{align}
By taking $n$ to infinity we obtain
\begin{align}
R(D)=\pi_1 H_b(p)+\pi_2 H_b(q)-H_b(D).\nonumber
\end{align}

\end{document}